\documentclass[11pt]{amsart}
\usepackage{appendix}
\usepackage{amssymb}
\usepackage{enumerate}
\usepackage{amsfonts}
\usepackage{amsmath}
\usepackage{mathrsfs}
\allowdisplaybreaks[1]

\newcounter{fig}

\newcommand{\myself}{\author{Gianluca Cassese}
                     \address{Universit\`{a} Milano Bicocca and University of Lugano}
                     \email{gianluca.cassese@unimib.it}
                     \curraddr{Department of Economics, Statistics and Management, 
                                  Building U7, Room 2097, via Bicocca 
                                  degli Arcimboldi 8, 20126 Milano - Italy}}

\newtheorem{theorem}{Theorem}
\theoremstyle{plain}

\newtheorem{corollary}{Corollary}
\newtheorem{assumption}{Assumption}
\newtheorem{definition}{Definition}
\newtheorem{example}{Example}

\newtheorem{lemma}{Lemma}

\newcommand{\Prob}{\mathbb{P}}

\newcommand{\Sim}{\mathscr{S}}

\newcommand{\B}{\mathfrak{B}} 
\newcommand{\A}{\mathscr{A}} 

\newcommand{\R}{\mathbb{R}} 
\newcommand{\N}{\mathbb{N}}

\newcommand{\RR}{\overline{\R}}

\newcommand{\quot}[1]{\textit{``#1''}}

\newcommand{\abs}[1]{\vert #1\vert} 
\newcommand{\dabs}[1]{\left\vert #1\right\vert} 

\newcommand{\dDer}[2]{\dfrac{d #1}{d #2}}

\newcommand{\net}[3]{\langle #1_{#2}\rangle_{#2\in #3} } 
\newcommand{\neta}[1]{\net{#1}{\alpha}{\mathfrak A}} 
\newcommand{\nnet}[3]{\langle #1\rangle_{#2\in #3} } 
 
\newcommand{\seq}[2]{\net{#1}{#2}{\mathbb{N}}} 
\newcommand{\sseq}[2]{\nnet{#1}{#2}{\mathbb{N}}} 
\newcommand{\seqn}[1]{\seq{#1}{n}} 
\newcommand{\sseqn}[1]{\sseq{#1}{n}}

\newcommand{\norm}[1]{\Vert #1\Vert}

\newcommand{\set}[1]{\mathbf{1}_{#1}}
\newcommand{\sset}[1]{\mathbf{1}_{\{#1\}}}

\newcommand{\emp}{\varnothing}

\newcommand{\iref}[1]{(\textit{\ref{#1}})}
\newcommand{\imply}[2]{\iref{#1}$\Rightarrow$\iref{#2}}

\newcommand {\X}{\mathscr X}

\newcommand {\Neg}{\mathcal N}
\newcommand {\Fun}{\mathfrak F}
\newcommand {\FA}{\mathfrak F(\mathfrak A)}
\newcommand {\TA}{\mathfrak A}
\newcommand {\FFA}{\mathfrak F_0(\mathfrak A)}
\newcommand {\K}{\mathscr K}

\newcommand {\La}{\mathscr L}
\newcommand {\M}{\mathscr M}

\begin{document} 
\title[Asset Pricing]{Asset Pricing in an Imperfect World}
\myself 
\date\today 
\thanks{I am grateful to Frank Riedel for discussion on the topic of the paper. }

\begin{abstract} 
In a model with no given probability measure, we consider asset pricing in the 
presence of frictions and other imperfections and characterize the property of 
coherent pricing, a notion related to (but much weaker than) the no arbitrage 
property. We show that prices are coherent if and only if the set of pricing 
measures is non empty, i.e. if pricing by expectation is possible. We then obtain 
a decomposition of coherent prices highlighting the role of bubbles. Eventually 
we show that under very weak conditions the coherent pricing of options 
allows for a very clear representation which allows, as in Breeden and 
Litzenberger \cite{breeden litzenberger}, to extract the implied probability. 
\end{abstract}

\maketitle

\section{Introduction} 
In this paper we study asset pricing for economies in which trading is prone 
to a wide variety of restrictions and costs and in which, in addition, agents 
need not possess the degree of sophistication required in order to assess 
uncertain outcomes via a probability function -- no matter how general we 
interpret this concept. Yet we obtain an exact relationship between prices 
and integrals and a representation of option prices in terms of the implied 
probability.

The assumption of an exogenously given probability contributes substantially 
to traditional financial models. In the first place, it sets the ambient space with
quantities being defined up to null sets rather than as a detailed list of 
characteristics, as in Arrow's notion of a contingent good. Second, in continuous 
time working with null sets is essential for a mathematically sound definition 
of the process of gains from trade and, more generally, for the construction of 
a rich enough set of admissible trading strategies so as to guarantee the opportunity 
of hedging many different derivatives. But even more importantly, the assumption of 
a given probability is crucial for the core principle of modern financial theory, i.e. 
risk neutral pricing. Although many an author inclines to believe that this basic 
principle rests on the simple tenet asserting that markets populated by rational 
economic agents cannot admit arbitrage opportunities, the proof of this claim, 
the \textit{fundamental theorem of asset pricing}, has long been a challenge 
for mathematical economists, from Kreps \cite{kreps} to Delbaen and 
Schachermayer \cite{delbaen schachermayer}. In fact it requires a much 
more stringent condition than absence of arbitrage in which probability is 
needed to induce an appropriate topology.

A convincing amount of empirical and experimental observations seems, however, 
to document attitudes of investors towards uncertainty which are deeply at odds 
with probabilistic rationality%
\footnote{
See \cite{MAFI} and references therein for a brief survey of the experimental
psychology literature and its relation to finance.
}. 
Moreover, some frequently observed facts, as reported e.g. by Lamont and Thaler 
in \cite{lamont thaler} and \cite{lamont thaler JEP}, or as embodied in the high ratio 
of option prices violating some basic no arbitrage condition, are often hard to 
reconcile with classical financial theory. Eventually, the tradition of subjective 
probability suggests that probabilities should be regarded more as the outcome 
of choice than as an input to it.  We take these considerations as the starting 
motivation for our choice in favor of an asset pricing model with no probability 
assumption, a step first made in \cite{MAFI} and then by a few papers which 
include 
\cite{riedel} and \cite{vovk}%
\footnote{
See the paper of Vovk \cite{vovk} for a general introduction to this topic and
for further references.
}.

Once this crucial choice is made a natural possibility to explore is whether
it is possible to retain the full fledged mathematical power of traditional 
models by detecting some reference probability that may employed with 
no additional assumption and without affecting the underlying economic 
model. This is clearly the case with a countable state space, in which any 
strictly positive probability may be introduced at no modeling cost. With a 
general state space, however, this is hardly possible as any probability 
would introduce many more null sets than the empty one, thus distorting 
the original financial model. Following this path it is then unavoidable to 
replace the assumption of a given probability with one guaranteeing some 
countable structure of the state space, as done in \cite{riedel} where the 
state space is a complete separable metric space and assets payoff are 
taken to be continuous.

We take here a radically different and fully general approach with no mathematical 
assumption, neither topological nor measure theoretic, to start with. Our \textit{a priori}
is rather a partial order providing a qualitative description of how uncertain outcomes 
are ranked by agents. We propose a reasonable set of axioms which cover several 
possible situations of interest, including the classical model which embeds thus in our
framework as a particularly interesting special case, another one being that
introduced in \cite{MAFI} in which agents base their choice on a class of
negligible events. We show in Theorem \ref{th risk} that a partial order 
satisfying our axioms may also be interpreted in terms of a corresponding 
coherent risk measure.

Another distinguishing feature of this paper is a rather general description of the market 
mechanism. Not only we consider bid/ask spreads and fixed trading costs but 
we also allow for several additional restrictions to trade. First, agents may be 
prevented from forming portfolios of arbitrarily large size -- it may thus not be 
possible to run money pumps, in case they exist. Arbitrage phenomena have 
as a consequence a minor impact on economic equilibrium and the no arbitrage 
principle looses part of its appeal. Second, not all portfolios may be shorted, and 
not just as a result of trading restrictions. In the presence of credit risk, taking 
long or short positions should be considered as two separate investments, given 
that the implicit level of risk depends on the reliability of the investor playing the 
short side. Third, we do not identify assets with their payoff so that it is possible 
to have two assets promising the same payoff at different prices. Fourth, the 
opportunity to invest in some asset may be available only if combined with other 
assets, e.g. with some collateral. The impact of margins when shorting options, 
an empirical fact accommodated in our model, has rarely been considered in 
asset pricing. Eventually, we don't assume the existence of a riskless asset but 
rather of a \textit{num\'eraire} whose non negative payoff is used as the discount 
factor. In contrast with the basic principles of neoclassical economics, the choice 
of the \textit{num\'eraire}, given the absence of a reference probability, is non 
neutral and an arbitrage opportunity arising with a given \textit{num\'eraire} 
may no longer be such with a different discount factor. 

In the framework outlined above we discuss three distinct notions of market 
rationality: coherence, efficiency and absence of arbitrage opportunities. We 
obtain in Theorem \ref{th coherent} a characterization of coherent prices in 
terms of a set of pricing measures. Our result is near in spirit to that obtained 
in the pioneering work of Jouini and Kallal \cite{jouini kallal} (in an $L^2(P)$ 
setting) but departs from it is several ways. First, pricing measures do not
apply but to claims with limited discounted losses; second, pricing measures
are just finitely additive. We obtain in Theorem \ref{th NFL} exact necessary 
and sufficient conditions for countable additivity which justify 
the conclusion that this additional property should be regarded more as a 
mathematical artifact than an economic implication. In Theorem \ref{th beta} 
we decompose coherent prices highlighting the role of bubbles. 
It is also possible to represent pricing measures with a capacity, similarly to 
what assumed in the work of Chateauneuf et al. \cite{chateauneuf}. The 
connection with capacities is also explored in Cerreia-Vioglio et al. \cite{cerreia}.

Of course, over the years several authors have investigated restrictions to 
trading similar to those considered here. Leaving aside the microstructure 
literature, in which transaction costs are the heart of the matter, the first 
papers have been those of Bensaid et al. \cite{bensaid} and, most of all, of 
Jouini and Kallal \cite{jouini kallal}. More recent papers include Bouchard 
\cite{bouchard}, Napp \cite{napp} (who first models trading restrictions 
via closed convex cones),  Jouini and Napp \cite{jouini napp} (who describe 
investments as cash flows with convex cone constraints and assume no 
\textit{num\'eraire}), Kabanov and Stricker \cite{kabanov stricker} (who consider 
very general forms of costs) and Schachermayer \cite{schachermayer}. With no 
claim to completeness, one should also mention the work of Amihud and Mendelson 
\cite{amihud}, Prisman \cite{prisman} and, most recently, Roux \cite{roux}.

This paper is structured as follows. In section \ref{sec markets} we describe 
markets and introduce the partial order needed to rank uncertain outcomes. In section 
\ref{sec coherence} we discuss the properties of coherence, efficiency and of 
absence of arbitrage and, in the following section \ref{sec coherent}, we obtain 
an explicit characterization of coherent prices in terms of pricing measures. In 
section \ref{sec decomposition} we develop a decomposition of coherent asset 
prices emphasizing the existence of bubbles. We then consider option markets 
in section \ref{sec option} where we prove a general representation for prices of 
convex derivatives, involving bubbles and an implicit pricing measure. Some 
indications for applied work are also given. Auxiliary results and some of the 
proofs will be found in the Appendix.

\subsection{Notation}

Throughout the paper we adopt the following mathematical symbols and conventions.
$\Fun(X)$ denotes the collection of real-valued functions on some space $X$ and 
$\Fun_0(X)$ designates those $f\in\Fun(X)$ whose support $\{x\in X:f(x)\ne0\}$ is a 
finite set. If $f\in\Fun(X)$ and if $-X\subset X$, the symbol $f^c$ will be used to 
denote the conjugate of $f$ defined as
$
f^c(x)=-f(-x)\qquad x\in X
$.
We set conventionally $0/0=0$, $\sup\emp=-\infty$, $\inf\emp=\infty$ and 
$\sum\emp=0$. $\RR$ denotes the extended real numbers.

When a given set $\Omega$ (resp. a family $\A$ of subsets of some set 
$\Omega$) is given, $ba$ (resp. $ba(\A)$) denotes the family of bounded, 
finitely additive set functions defined on all subsets of $\Omega$ (resp. on 
all sets in $\A$). $\Sim(\A)$ denotes the family of $\A$ simple functions. 
When $\mu\in ba$ and $f\in\Fun(\Omega)$ we define its integral as
\begin{equation}
\int fd\mu=\lim_n\int [(f^+\wedge n)-(f^-\wedge n)]d\mu
\end{equation}
if such limit exists in $\RR$, or else $\int fd\mu=\infty$. It is easily seen that 
this definition coincides with \cite[III.2.17]{bible} whenever $f$ is $\mu$ 
integrable, i.e. $f\in L^1(\mu)$. 

\section{Markets, Prices, Investors}
\label{sec markets}
Assets traded on the market are identified with a \quot{ticker}, $\alpha\in\TA$, 
and are associated with a corresponding payoff, $X(\alpha)$. The latter is 
modeled simply as a function on some given space $\Omega$, i.e. an 
element of $\Fun(\Omega)$. As discussed in the Introduction, no mathematical 
structure is imposed on the set of traded payoffs. Although it is natural to 
interpret $\Omega$ as the sample space and $X(\alpha)$ as a random 
quantity (which makes our model look intrinsically static) we may
as well choose $\Omega=S\times\R_+$, with $S$ the sample space and 
$\R_+$ the time domain -- and thus give to our construction  a full fledged 
dynamical structure.

\subsection{Trading Strategies}
Investors trade claims by taking a finite number of either long or short positions, 
in respect of the restrictions imposed by the market. A trading rule is then just 
an element of the space $\FFA$. The trading rule which consists solely of one 
unit of the claim $\alpha\in\TA$ will be denoted by $\delta_\alpha$. To each 
trading rule $\theta$ corresponds the final gain 
\begin{equation}
X(\theta)=\sum_{\alpha\in\TA}\theta(\alpha)X(\alpha)
\end{equation}
Of course, $X(\delta_\alpha)=X(\alpha)$. 

Inspired by real markets, one may imagine several restrictions to asset trading, 
further to the constraints of respecting the balance of budget and of forming 
\textit{finite} portfolios. These include short selling prohibitions or margin 
requirements and others that ultimately aim at enforcing some 
form of bound to losses. The symbol $\Theta$, that denotes hereafter the set 
of all admissible trades, specifies all restrictions to asset trading. We assume 
the following:

\begin{assumption}
\label{ass Theta}
$\Theta$ is a convex subset of $\FFA$ containing the origin.
\end{assumption}

Under Assumption 1 investors need not be permitted to take positions of either sign,
long or short. This is consistent with the restriction to losses recalled above. Moreover,
investors may encounter restrictions in the choice of the scale of the investment. On 
this point we depart significantly from much of the literature on asset pricing with or 
without transaction costs, see e.g. \cite{jouini kallal}, \cite{kabanov stricker} or 
\cite{luttmer}. A possible relaxation of this restriction (on which we shall return) is 
to allow $\lambda\theta\in\Theta$ for all $\lambda\ge0$ whenever $\theta\in\Theta$ 
satisfies $X(\theta)\ge0$. A major implication of this is that arbitrage opportunities, 
when available, may not have a disruptive impact on market equilibrium and the no 
arbitrage principle, as a consequence, may no longer be crucial. Eventually, we do not 
require that $\delta_\alpha\in\Theta$, i.e. that each asset may be traded individually 
due, e.g., to the requirement of putting up margins when taking positions on derivative
markets.

\subsection{Prices and Costs}

For each $\alpha\in\TA$ we denote by $q^a(\alpha)$ and $q^b(\alpha)$ its \textit{ask} 
and \textit{bid} price respectively. 

\begin{assumption}
\label{ass b/a}
The functions $q^a,q^b\in\FA$ are such that $q^a(\alpha)\ge q^b(\alpha)$ 
for all $\alpha\in\TA$.
\end{assumption}

We highlight that in our model financial prices are not defined as functions of the 
asset payoff but depend rather of its name. This apparently innocuous detail makes 
our approach compatible with some pricing anomalies reporting that the trading 
of one same asset at different market locations or simply under different names, 
may produce different prices (see the examples on close-end funds or of twin 
stocks reported in Lamont and Thaler \cite{lamont thaler JEP}).

In order to form a given trading strategy $\theta\in\Theta$ an investor pays an 
ask price for each long position and earns a bid price for each short one. The
corresponding cost amounts thus to
\begin{equation}
\label{q}
q(\theta)
=
\sum_{\alpha\in\TA}\left[\theta(\alpha)^+q^a(\alpha)-\theta(\alpha)^-q^b(\alpha)\right]
\qquad\theta\in\Theta
\end{equation}
It is clear that $q(\delta_\alpha)=q^a(\alpha)$ and, if $-\delta_\alpha\in\Theta$, that 
$q^c(\delta_\alpha)=-q(-\delta_\alpha)=q^b(\alpha)$.

The basic assumption behind \eqref{q} is that each position in a portfolio is priced
separately. This is in accordance with trading anonymity prevailing in specialist markets
but is perhaps not an adequate description of OTC trading. Options markets, on
which we shall focus in the last sections, are quite well represented by \eqref{q}
at least for orders which fall below the size limits of the market maker. Large 
orders, instead, are in general processed on a separate track and the price is
set \textit{ad hoc}.

Market frictions include, further to the bid/ask spread, also some fixed costs, such 
as brokerage fees, in the form of a sunk payment due to have access to the 
market. Given that on each market investors may trade more than one asset, we 
may thus think of markets as a partition $\mathfrak M$ of subsets of $\TA$ and for 
each $M\in\mathfrak M$ we designate by $c(M)\ge0$ the corresponding fixed cost. 
For example, all options on a given underlying are traded on the same market, 
independently of the strike or maturity so that any option strategy will involve the 
same fees. Thus the fixed cost associated with an investment strategy is
\begin{equation}
\label{fee}
c(\theta)
=
\sum_{\{M\in\mathfrak M:\ \sup_{\alpha\in M}\abs{\theta(\alpha)}>0\}}c(M)
\qquad\theta\in\Theta
\end{equation}
A realistic modeling of fixed trading costs turns out to be quite difficult due to
their extremely various nature%
\footnote{
As a matter of fact transaction fees tend to be stepwise increasing with the order size
rather then fixed.
}.
\eqref{fee} is just one possible model.

The total cost associated with a trading strategy $\theta$ amounts to
\begin{equation}
\label{cost}
t(\theta)=q(\theta)+c(\theta)
\qquad
\theta\in\Theta
\end{equation}

It is easy to deduce from \eqref{q} and \eqref{cost} some elementary properties:

\begin{lemma}
\label{lemma subadditive}
The functional $q\in\Fun(\Theta)$ defined in \eqref{q} is (i) positively homogeneous, 
(ii) subadditive and (iii) such that 
\begin{equation}
\label{anonymity}
q(f+g)=q(f)+q(g)
\qquad\text{for all}\quad
f,g\in\Theta
\quad\text{with}\quad
fg\ge0
\end{equation}
The functional $t\in\Fun(\Theta)$ defined in \eqref{cost} is subadditive and 
satisfies $t(0)=0$.
\end{lemma}

In many a paper on asset pricing with frictions, starting with the seminal paper of
Jouini and Kallal \cite{jouini kallal}, subadditivity is the only distinguishing feature 
characterizing the existence of bid/ask spreads. Another paper following this choice 
is that of  Luttmer \cite{luttmer}. In Chateauneuf et al. 
\cite{chateauneuf}, the pricing functional is represented via a capacity and is then 
not only subadditive but even comonotonic, a property somehow akin to 
\eqref{anonymity}. In these papers an explicit description of the costs of trading 
is omitted (a remarkable exception is \cite{kabanov stricker}) and in so doing, we 
claim, one risks to miss important details of the price mechanism and to mix up 
effects that may actually originate from different sources, e.g. bounded rationality 
or restrictions to market participation. We will refer to property \eqref{anonymity}
as \textit{anonymity} and, although in the following Theorems \ref{th coherent} 
through \ref{th beta} this plays no role, it will be important when dealing with option prices 
for which, we believe, it is perfectly adequate.

\subsection{The \textit{Num\'eraire} Asset}
Financial models (with the noteworthy exception of \cite{jouini napp}) commonly 
assume the existence of a riskless asset, often interpreted as a bond, that may 
equally well serve the purpose of borrowing or  lending. This assumption plays 
three distinct roles. First, it enables agents to move wealth back and forth in time 
in a safe way and thus provides a firm basis to define the present value of future 
wealth. Second, it allows to identify explicitly the \textit{num\'eraire} of the 
economy, removing the arbitrariness that arises whenever several assets may 
play that same role. Third, if the investment in the bond is unrestricted then this 
asset plays a residual role in portfolio models, guaranteeing the effectiveness 
of portfolio constraints.

This assumption is however not only strongly counterfactual but more troublesome 
than it appears at first sight. First, if the bond is not fully free of risk but just evolves 
in a predictable way (as is often the case in continuous time models) then the role 
of the discounting asset is no longer neutral as its implicit risk entwines with the one 
originating from the underlying asset. El Karoui and Ravanelli \cite{el karoui ravanelli} 
discuss this point at length and show  that risk measures are affected by discounting 
in a significant way, when the discount factor is risky. Moreover, in equilibrium models, 
such as those considering the role of \textit{noise trading} (see \cite{DSSW} or 
\cite{shleifer vishny}), the riskless nature of the bond may not survive Walras law 
unless its elasticity of supply is infinite%
\footnote{See the criticisms to this assumption made by Lowenstein and Willard
\cite{lowenstein willard 2006}.}. 
Eventually if investors are prone to credit risk one should consider borrowing and lending 
as two different financial contracts, given that the final payoff is ultimately a function of 
the reliability of the two intervening parts. 

We summarize the preceding discussion in the following:

\begin{assumption}
\label{ass safe}
There exists $\alpha_0\in\TA$ such that 
(i) $1\ge X(\alpha_0)>0$, 
(ii) if $\theta\in\Theta$ then $\theta+\lambda\delta_{\alpha_0}\in\Theta$ if and 
only if $\lambda\ge0$,
(iii) $c(\delta_{\alpha_0})=0$,
(iv) $q(\alpha_0)>0$.
\end{assumption}


We shall refer to $\alpha_0$ as the \textit{num\'eraire} asset and simplify  
$\delta_{\alpha_0}$, $X(\alpha_0)$ and $q(\alpha_0)$ as $\delta_0$, $X_0$ 
and $q_0$. Assumption \ref{ass safe}.(\textit i) is fairly general to allow for virtually 
all sorts of dynamics but it excludes the occurrence of default with no recovery 
value. If one visualizes the \textit{num\'eraire} asset as government bonds one 
may perhaps consider this restriction not too far from reality, given that in market 
economies government bonds have always been redeemed at some positive value. 
Alternatively, identifying the \textit{num\'eraire} asset with the bank account one 
may argue, likewise, that the bank account is in most countries assisted by some 
form of deposit insurance, maybe just in the form of the role of lender of last 
resort played by the Central Bank. As in the real world, in our model investors are 
unrestricted in deciding the amount to invest in the \textit{num\'eraire} asset, 
but they cannot take negative positions as this would more appropriately be considered a
different asset, as argued above. Property (\textit{iii}) may perhaps be seen as
the outcome of competition among banks, while (\textit{iv}) justifies referring
to $\alpha_0$ as the \textit{num\'eraire}.

We define normalized payoffs as
\begin{equation}
\label{normalized}
\bar X(\theta)=X(\theta)\slash X_0
\end{equation}

\subsection{Stochastic ordering}
An agent's decision to invest in a given trading strategy $\theta\in\Theta$ is motivated, 
we assume, by the payoff $X(\theta)$ that it generates. However, a full description of 
this quantity for each possible future state $\omega\in\Omega$ is not necessarily a 
correct model of choice as economic agents often do not regard future outcomes as 
\textit{functions} but rather as \textit{equivalence classes}. This is clearly the case 
in expected utility theory and, more generally, in all probability models. Equivalence, 
however, is not only the outcome of an accurate 
probabilistic assessment -- as the classical model implicitly suggests -- but it often 
emerges from the inability of individuals to fully compare events or from their attitude 
to focus attention on scenarios selectively, a fact often documented in empirical decision 
theory and experimental psychology.

These remarks suggest to treat stochastic order as an explicit a priori
of our model and to model it via a binary relation $\ge_*$ on $\Fun(\Omega)$%
\footnote{
We reserve the notation $f\ge g$, $f\vee g$, $f\wedge g$ or $\abs f$ to 
pointwise ordering.
}.
We assume to this end:

\begin{assumption}
\label{ass preorder}
The binary relation $\ge_*$ on $\Fun(\Omega)$ is reflexive, transitive
and satisfies:\\
\begin{tabular}{ll}
(TRIV)& $0\not\ge_*1$;\\
(CONE)&$f_i\ge_*g_i$ and $a_i\in\R_+$ for $i=1,2,$ imply
$a_1f_1+a_2f_2\ge_*a_1g_1+a_2g_2$;\\
(CERT)&$f\ge0$ implies $f\ge_*0$;\\
(APPR)&$f+2^{-n}\ge_*0$ for $n=1,2,\ldots$ implies $f\ge_*0$;\\
(REST)&$f\ge_*0$ and $A\subset\Omega$ imply $f\set A\ge_*0$.
\end{tabular}
\end{assumption}

(\textit{TRIV}) prevents trivial situations in which all bounded functions are 
equivalent to $0$.
(\textit{CONE}) guarantees that non negative elements define a convex cone, a 
crucial property when forming portfolios. (\textit{CERT}) states that the order 
does not contradict certainty; (\textit{APPR}) establishes, more interestingly, 
that a quantity that may be approximated by one considered as non negative by 
an amount which is arbitrarily small for all practical purposes should itself be 
considered as non negative. Eventually, (\textit{REST}) implies that non negativity 
is a global assessment and is preserved when passing to subsets.

We also write $f>_*g$ (resp. $f=_*g$) when $f\ge_* g$ but $g\not\ge_* f$ 
(resp. and $g\ge_*f$) and define
\begin{equation}
\label{f*}
f_*=\sup\{\alpha\in\R:f\ge_*\alpha\}
\quad\text{and}\quad
f^*=-(-f)_*
\end{equation}

A partial order on $\Fun(\Omega)$ satisfying Assumption \ref{ass preorder} 
will be referred to as a regular stochastic order. The following are examples 
such order.

\begin{example}[Certainty and probability]
\label{ex obvious}
Define $f\succeq_0g$ and $f\succeq_Pg$ to mean $f\ge g$ and $P(f\ge g)=1$
if, given a (countably additive) probability $P$.  Both $\succeq_0$ and $\succeq_P$
are regular stochastic orders, the first often being referred to as zero-th order
stochastic dominance. First order stochastic dominance is not a regular stochastic
order as it fails to satisfy (\textit{CONV}). Our setting therefore covers the case of
certainty and probabilistic sophistication.
\end{example}

\begin{example}[Qualitative probability]
\label{ex qualitative}
In his pioneering work, de Finetti \cite{de finetti} introduced the idea of modeling 
the qualitative judgment \quot{event A is more likely than B} as a binary relation, 
$A\succeq B$ satisfying the axioms: 
(\textit a) completeness, 
(\textit b) transitivity, 
(\textit c) $\Omega\succeq A\succeq\emp$ for all events $A$ and 
(\textit d) if  $C\cap(A\cup B)=\emp$ then $A\succeq B$ if and only 
if $A\cup C\succeq B\cup C$. One may define 
\begin{equation}
\label{qualitative preorder}
f\succeq_{dF}g
\quad\text{if and only if}\quad
\{f-g\le-\eta\}\preceq\emp
\qquad\text{for all}\quad 
\eta>0
\end{equation}
It is easily seen that, adding the axiom (e) $\emp\not\succeq\Omega$, then 
$\succeq_{dF}$ is a regular stochastic order exactly because the collection 
$\{A\subset\Omega:A\preceq\emp\}$ is an ideal of subsets of $\Omega$. A 
generalization of this idea was introduced in \cite{MAFI} where an ideal $\Neg$ 
of so-called negligible events (not including $\Omega$) was taken as a primitive 
and a corresponding order $\succeq_\Neg$ was defined as in \eqref{qualitative 
preorder}, i.e.
\begin{equation}
f\succeq_\Neg g
\quad\text{if and only if}\quad
\{f-g\le-\eta\}\in\Neg
\qquad\text{for all}\quad 
\eta>0
\end{equation}
\end{example}

Regular stochastic orders induced by ideals of sets can be characterized as follows:

\begin{lemma}
\label{lemma neg}
A regular stochastic order $\succeq$ is induced by an ideal of subsets of $\Omega$
if and only if there exists a weak$^*$ compact set $\Prob_\succeq\subset\Prob$
such that $\sup\{\mu(A):\mu\in\Prob_\succeq\}\in\{0,1\}$ for all $A\subset\Omega$
and
\begin{equation}
\label{P}
f\succeq g
\quad\text{is equivalent to}\quad
\inf_{\mu\in\Prob_\succeq}\int(f-g)d\mu\ge0
\end{equation}
\end{lemma}

\begin{proof}
Assume that $\succeq$ is a regular stochastic order induced by an ideal
$\Neg$ and let
\begin{equation}
\Prob_\succeq
	=
\{\mu\in\Prob:\mu(A)=0\text{ for all }A\in\Neg\}
\end{equation}
Suppose that $\mu(A)>0$ for some $A\subset\Omega$ and $\mu\in\Prob_\succeq$.
Then $\mu_A$, the conditioning of $\mu$ to $A$, is again an element
of $\Prob_\succeq$ and $\mu_A(A)=1$. Thus
$\sup\{\mu(A):\mu\in\Prob_\succeq\}\in\{0,1\}$ for all $A\subset\Omega$.
Take $A\not\in\Neg$. The set $\{a\set{A^c}:a\ge0\}$ is a convex cone and 
$a\set{A^c}\succeq1$ would imply $0\succeq\set A$, a contradiction. By
Theorem \ref{th cone} there exists $\mu\in\Prob_\succeq$ such that 
$\mu(A^c)=0$. This proves that $\sup\{\mu(A):\mu\in\Prob_\succeq\}=1$
if and only if $A\notin\Neg$. Suppose that $f\succeq0$. Fix $\eta>0$ and observe 
that $-\eta\sset{f<-\eta}\ge f\sset{f<-\eta}\succeq0$, by (\textit{CERT}) and 
(\textit{REST}). It follows from (\textit{CONV}) that $\{f<-\eta\}\in\Neg$ so that
\begin{align*}
\inf_{\mu\in\Prob_\succeq}\int fd\mu
	=
\inf_{\mu\in\Prob_\succeq}\int_{\{f\ge-\eta\}} fd\mu
	\ge
-\eta
\end{align*}
and thus that $\inf_{\mu\in\Prob_\succeq}\int fd\mu\ge0$. Conversely, let 
$\inf_{\mu\in\Prob_\succeq}\int fd\mu\ge0$. If 
$\sup_{\mu\in\Prob_\succeq}\mu(f<-\eta)=1$ for some $\eta>0$, then, by weak$^*$
compactness there would exist $\mu_0\in\Prob_\succeq$ such that
$\mu_0(f<-\eta)=1$ and so $\int d\mu_0\le-\eta$, a contradiction.
Thus necessarily $\sup_{\mu\in\Prob_\succeq}\mu(f<-\eta)=0$,
i.e. $\{f<-\eta\}\in\Neg$ so that $f\succeq0$. Eventually, if $\Prob_\succeq$ 
has the above properties and \eqref{P} is taken as a definition of $\succeq$, define 
$\Neg=\{A\subset\Omega:\sup_{\mu\in\Prob_\succeq}\mu(A)=0\}$. It is
easily seen that $\Neg$ is an ideal of subsets of $\Omega$ and that
$\Omega\notin\Neg$. Moreover, if $\inf_{\mu\in\Prob_\succeq}\int(f-g)d\mu\ge0$
it must be that $\sup_{\mu\in\Prob_\succeq}\mu(f-g\le-\eta)=0$ for all
$\eta>0$ and so $\succeq$ is indeed a regular stochastic order induced by
an ideal.
\end{proof}

Despite the natural interpretation of $\ge_*$ as an element of subjective choice, 
it may well be interpreted in more applied terms. Let us recall \cite[Definition 3.1]
{el karoui ravanelli} that a map $\rho:\Fun(\Omega)\to\RR$ is a coherent, cash 
subadditive risk measure whenever $\rho$ is
(\textit a) positively homogeneous, 
(\textit b) subadditive, 
(\textit c) inversely monotone (i.e. $f\le g$ implies $\rho(f)\ge\rho(g)$) and 
(\textit d) cash-subadditive, (i.e. $\rho(f+\alpha)\ge\rho(f)-\alpha$ when $\alpha\in\R_+$). 
If, in addition, $\rho$ satisfies
(\textit e) $\rho(f)=\rho(f\wedge0)$,
then we refer to it as a \textit{loss measure}. It may be easily proved that
if $\rho$ is a coherent, cash subadditive risk measure, then 
$\hat\rho(f)=\rho(f\wedge0)$ is a coherent, cash subadditive loss measure.

\begin{theorem}
\label{th risk}
$\succeq$ is a regular stochastic order if and only if there exists a coherent, cash 
subadditive loss measure $\rho$ such that $\rho(-1)>0$ and 
\begin{equation}
\label{succ}
f\succeq g
\quad\text{if and only if}\quad
\rho(f-g)\le0
\qquad\text{for all }
f,g\in\Fun(\Omega)
\end{equation}
\end{theorem}

\begin{proof}
Let $\succeq$ be a regular stochastic order. Define
\begin{equation}
\sigma_\succeq(f)=\inf\{\beta\in\R:(f\wedge0)+\beta\succeq0\}
\end{equation}
$\sigma_\succeq$ is subadditive and positively homogeneous by (\textit{CONE})
and inversely monotone by (\textit{CERT}). $\sigma_\succeq(-1)>0$ is obvious
given (\textit{TRIV}). To prove cash subadditivity we may restrict to the case of 
$f\in\Fun(\Omega)$ and $a>0$ such that $\sigma_\succeq(f+a)<\infty$ i.e. such 
that $((f+a)\wedge0)+\beta\succeq0$ for some $\beta\in\R$. But then
\begin{align*}
f\wedge0+(a+\beta)
	\ge
f\wedge -a+(a+\beta)
	\succeq
0
\end{align*}
so that, by (\textit{CERT}), $a+\beta\ge\sigma_\succeq(f)$ i.e. 
$\sigma_\succeq(f+a)\ge\sigma_\succeq(f)-a$. To see that $\succeq$ is 
related with $\sigma_\succeq$ via \eqref{succ}, observe that
$\sigma_\succeq(f)\le0$ is equivalent to $(f\wedge0)+2^{-n}\succeq0$ 
i.e., by (\textit{APPR}), to $(f\wedge0)\succeq0$ and this in turn to $
f\succeq0$ by (\textit{CERT}) and (\textit{REST}).

Conversely, assume that $\rho$ is a coherent, cash subadditive loss measure and 
define $\succeq$ via \eqref{succ}. Suppose that $f_i\succeq g_i$ and that 
$a_i\in\R_+$ for $i=1,2$. Then,
\begin{align*}
\rho(a_1(f_1-g_1)+a_2(f_2-g_2))
	\le
a_1\rho(f_1-g_1)+a_2\rho(f_2-g_2)
	\le
0
\end{align*}
so that $a_1f_1+a_2f_2\succeq a_1g_1+a_2g_2$ and (\textit{CONE}) holds.
(\textit{CERT}) is clear while (\textit{REST}) follows from
\begin{align*}
\rho(f\set A)
	=
\rho(\set A (f\wedge 0))
	\le
\rho(f\wedge 0)
	=
\rho(f)
\end{align*}
Eventually, by
$
\rho(f+2^{-n})
	\ge
\rho(f)-2^{-n}
$
we conclude that $f+2^{-n}\succeq0$ for all $n$ implies $f\succeq0$.
\end{proof}

We will write 
\begin{equation}
\label{N*}
\Neg_*
	=
\left\{A\subset\Omega:0\ge_*\set A\right\},
	\qquad
\B_*
	=
\{f\in\Fun(\Omega):\eta\ge_*\abs f\text{ for some }\eta>0\}
\end{equation}

\begin{equation}
\label{P*}
\Prob_*=\{\mu\in\Prob:\mu(A)=0\text{ for all }A\in\Neg_*\}
\end{equation}•
and, when $\A$ is an algebra of subsets of $\Omega$ containing $\Neg_*$,
\begin{equation}
\label{ba*}
ba_*(\A)=\{\lambda\in ba(\A):\lambda(N)=0\text{ for all }N\in\Neg_*\}
\end{equation}

A noteworthy property of the regular stochastic order $\ge_*$ is the following:
\begin{align}
\label{prop}
f\ge_*0
\text{ and } 
b\in\B_+
\qquad\text{imply}\qquad
fb\ge_*0
\end{align}
a fact that follows from (\textit{REST}) when $b$ is simple and extends
to the more general case by (\textit{APPR}) and uniform convergence.

Property \eqref{prop} and Assumption \ref{ass safe}.(\textit i) have an interesting
economic implication, namely that $\bar X\ge_*0$ implies $X\ge_*0$ but the converse 
need not be true. This special feature of our model highlights the role of discounting 
in the overall level of risk. The statement $X\ge_*0$, in fact, does not exclude losses 
but rather implies $\{X<-\eta\}\in\Neg_*$ for every $\eta>0$, i.e. that losses may be 
considered as arbitrarily small. The statement $\bar X\ge_*0$  means, on the other 
hand, that losses from $X$ may be hedged 
by investing an arbitrarily small amount in the \textit{num\'eraire} asset. If the 
\textit{num\'eraire} does not guarantee a minimum 
payoff, the losses associated with $\theta$, although small, may require a potentially 
unbounded amount of such asset in order to be hedged. The problem arises whenever 
losses from a  portfolio occur jointly with a drop in the value for $X_0$, as is often 
the case during financial crises. The risk management aspects of the choice of 
the \textit{num\'eraire} are also discussed by El-Karoui and Ravanelli 
\cite{el karoui ravanelli} and Filipovic \cite{filipovic}.

\section{Coherence, Efficiency and Arbitrage}
\label{sec coherence}
The basis of financial economics is the tenet that markets populated
by rational agents do not permit arbitrage opportunities. However, if there is 
agreement on this general statement, its translation into a convenient 
mathematical notion is much less uncontroversial. Definitions vary from one another 
mainly for the ambient space adopted and, since Harrison and Kreps \cite{harrison kreps},
the choice has traditionally been some $L^p(P)$ space, for a given exogenous probability 
measure $P$. We rather propose here the following definitions:

\begin{definition}
\label{def coherent}
A functional $\phi\in\Fun(\Theta)$ is said to be coherent (with the no arbitrage principle) 
if $\bar X(\theta)\ge_*0$ implies $\phi(\theta)\ge0$; $\phi$ is said to be
efficient if $\theta,\theta'\in\Theta$ and $\bar X(\theta)\ge_*\bar X(\theta')$ imply
$\phi(\theta)\ge\phi(\theta')$. Moreover, $\theta\in\Theta$ is an arbitrage opportunity 
for $\phi$ if
\begin{equation}
\label{arbitrage}
\bar X(\theta)\ge_*0
\quad\text{but}\quad
\phi(\theta)\le0
\end{equation} 
and at least one of the two inequalities is strict. 
\end{definition}

Although with a linear pricing rule and no restriction to short selling, coherence
and efficiency are equivalent properties, in the more general case treated
here coherence is weaker than efficiency nor does it guarantee absence of 
arbitrage \textit{per se}. Coherent pricing does not exclude that an 
investment which yields a strictly positive (discounted) payoff is sold for free. 
Coherence is thus a rather weak and basic financial property and we shall investigate 
it in depth. 

The inequality $\bar X(\theta)\ge_*0$ may be rephrased in terms of the minimal 
margin $M_*(\theta)$ to be invested in the \textit{num\'eraire} asset in order to 
hedge losses away (if possible). Formally,
\begin{equation}
\label{hedge}
M_*(\theta)
	=
\inf\{\eta>0:\bar X(\theta+\eta\delta_0)\ge_*0\}
	=
-(\bar X(\theta)_*\wedge0)
\end{equation}
By Assumption \ref{ass safe}.(\textit{ii}), $\theta+M_*(\theta)\delta_0\in\Theta$
if and only if $M_*(\theta)<\infty$ or, equivalently, if $\theta$ belongs to the set of 
hedgeable strategies
\begin{equation}
\label{Theta*}
\Theta_*=\left\{\theta:\bar X(\theta)_*>-\infty\right\}
\end{equation}
In fact regulated markets do not allow investors to enter positions with 
unlimited potential losses so that $\Theta_*$ is often considered as the set of 
all \textit{reasonable} investment strategies -- see \cite{delbaen schachermayer} 
where a condition akin to $\theta\in\Theta_*$ is the basis for the concept of 
\textit{free lunch with vanishing risk}.

Observe that $M_*(\theta+\alpha\delta_0)\ge M_*(\theta)-\alpha$ although
the basic intuition used by El-Karoui and Ravanelli to justify cash subadditivity 
(namely that the discount factor is less than 1, see \cite[p. 568]{el karoui ravanelli}) 
does not apply as we do not impose $X_0\ge_*1$. Moreover,
$M_*(\theta+M_*(\theta)\delta_0)=0$. 

\begin{theorem}
\label{th arbitrage}
Let $t$ be the total cost functional $t$ defined in \eqref{cost}. Then,
\begin{enumerate}[(i)]
\item
$t$ is coherent if and only if
\begin{equation}
\label{coherent t}
t(\theta)+q_0M_*(\theta)\ge0
\qquad\text{for all }
\theta\in\Theta_*
\end{equation}
\item
if $t$ is convex then it is coherent if and only if  there exists $\mu\in\Prob_*$ 
such that
\begin{equation}
\label{pricing t}
t(\theta)\ge q_0\int\left(\bar X(\theta)\wedge0\right)d\mu
\qquad\text{for all}\quad
\theta\in\Theta_*
\end{equation}
\item
$t$ admits no arbitrage opportunity if and only if it satisfies
\begin{equation}
\label{NA t}
t(\theta)+q_0M_*(\theta)>0
\qquad\text{for all }
\theta\in\Theta_*
\text{ such that}\qquad
\bar X(\theta)^*+M_*(\theta)>0
\end{equation}
\end{enumerate}
\end{theorem}

\begin{proof}
In studying coherence we are obviously entitled to restrict attention to $\Theta_*$.
For each $\theta\in\Theta_*$, let $\theta'=\theta+M_*(\theta)\delta_0\in\Theta$.
By Theorem \ref{th risk} and Assumption \ref{ass safe}.(\textit{iii}),
\begin{align}
\label{X+eta}
\bar X(\theta')
	&=
\bar X(\theta)+M_*(\theta)
	\ge_*0
&t(\theta')
	&=
t(\theta)+q_0M_*(\theta)
&\bar X(\theta')^*
	&=
\bar X(\theta)^*+M_*(\theta)
\end{align}
(\textit{i}). If $t$ is coherent then, $0\le t(\theta')=t(\theta)+M_*(\theta)q_0$ 
and \eqref{coherent t} holds. If, conversely, $\bar X(\theta)\ge_*0$, i.e. 
$M_*(\theta)=0$, then \eqref{coherent t} implies $t(\theta)\ge0$ so that $t$ is 
coherent. 

(\textit{ii}). Assume that $t$ is convex. Consider the sets
$
\mathscr H_0
=
\left\{\left(\bar X(\theta)\wedge0\right)q_0-t(\theta):
\theta\in\Theta\right\}
$
and 
\begin{equation*}
\mathscr H
=
\left\{f\in\Fun(\Omega):
f^-\in\B,\ 
\lambda h\ge_*f\text{ for some }h\in\mathscr H_0,
\text{ and }\lambda\ge0\right\}
\end{equation*}
If $\lambda_1,\ldots,\lambda_N>0$ and $\theta^1,\ldots,\theta^N\in\Theta$ then
$
\sum_{n=1}^N\lambda_n\left[(\bar X(\theta^n)\wedge0)q_0-t(\theta^n)\right]
\le
\lambda\left[(\bar X(\theta)\wedge0)q_0-t(\theta)\right]
$
with $\lambda=\sum_{n=1}^N\lambda_n$ and 
$
\theta
=
\sum_{n=1}^N(\lambda_n/\lambda)\theta^n\in\Theta
$.
Moreover, $h_n\ge_*f_n$ for $n=1,\ldots,N$ implies 
$\sum_{n=1}^Nh_n\ge_*\sum_{n=1}^Nf_n$. Thus $\mathscr H$ is a convex 
cone of uniformly lower bounded functions which, by \eqref{coherent t}, contains 
no element $f\ge_*1$. By Theorem \ref{th cone} in the Appendix there exists 
$\mu\in\Prob_*$ such that 
$$
\mathscr H\subset L^1(\mu)
\quad\text{and}\quad
\sup_{f\in\mathscr H}\int fd\mu\le0
$$
Moreover, if $\theta\in\Theta_*$ and $f=(\bar X(\theta)\wedge0)q_0-t(\theta)$ then
\begin{align*}
0
	\ge
\int_{\{\bar X(\theta)>\bar X(\theta)_*-\eta\}}fd\mu
	=
q_0\int_{\{\bar X(\theta)>\bar X(\theta)_*-\eta\}}(\bar X(\theta)\wedge0)d\mu-t(\theta)
	=
q_0\int(\bar X(\theta)\wedge0)d\mu-t(\theta)
\end{align*}
which proves the direct implication. The converse follows from the inequality
\begin{align*}
\int(\bar X(\theta)\wedge0)d\mu
	=
\int_{\bar X(\theta)\ge \bar X(\theta)_*-\eta}(\bar X(\theta)\wedge0)d\mu
	\ge
(\bar X(\theta)_*\wedge0)-\eta
\end{align*}
$q_0>0$ and \eqref{coherent t}.

(\textit{iii}). If $\bar X(\theta)^*+M_*(\theta)>0$ then, by \eqref{X+eta}, $\theta'$ 
is an arbitrage opportunity unless $t(\theta)+q_0M_*(\theta)>0$. \eqref{NA t} is thus 
necessary for absence of arbitrage. Conversely, choose $\theta\in\Theta_*$ and fix 
$\varepsilon\ge0$ such that $\bar X(\theta)^*+M_*(\theta)+\varepsilon>0$. By 
assumption, $\theta_\varepsilon=\theta+[\varepsilon+M_*(\theta)]\delta_0\in\Theta$ 
and
$
\bar X(\theta_\varepsilon)^*+M_*(\theta_\varepsilon)
>0$.
If \eqref{NA t} holds then
\begin{align*}
0
<
t(\theta_\varepsilon)+q_0M_*(\theta_\varepsilon)
=
t(\theta)+[\varepsilon+M_*(\theta)]q_0
\end{align*}
so that $t(\theta)+M_*(\theta)q_0>-\varepsilon q_0$ for all $\varepsilon>0$. Thus
$t(\theta)+M_*(\theta)q_0\ge0$ for all $\theta\in\Theta_*$ and $t$ is coherent.
If $\bar X(\theta)_*\ge0$ and $\bar X(\theta)>_*0$, then \eqref{NA t} implies 
$t(\theta)=t(\theta)+M_*(\theta)q_0>0$ so that $t$ admits no arbitrage opportunity.
\end{proof}

The representation \eqref{pricing t}, although quite manageable, relies crucially 
on convexity, a key property which is hard to justify based on the available empirical 
evidence which indicates, contrariwise, that fixed costs increase less than proportionally.
This conclusion suggests that a more interesting representation may require to focus
on the pricing functional, not including fixed costs.

\section{Coherent Pricing}
\label{sec coherent}

Since the early work of Bensaid et al \cite{bensaid} it is known that many properties
of asset prices are revealed by the super hedging functional and our 
model is no exception. We adapt this concept in defining the following extended
real valued functional:
\begin{equation}
\label{pi}
\pi(f)=\inf\left\{\lambda q(\theta):
\lambda \bar X(\theta)\ge_*f,\lambda\ge0,\theta\in\Theta\right\}
\qquad f\in\mathfrak F(\Omega)
\end{equation}
Clearly, $\pi(1)\le q_0$ and $\pi^c(1)\ge0$; if, in addition, $q$ is coherent, then 
$\pi(0)=0$ and $\pi^c(f)\le\pi(f)$ for all $f\in\Fun(\Omega)$ -- see Lemma 
\ref{lemma pi}. But even assuming coherence we cannot exclude the somehow 
abnormal situations $\pi(1)=0$ and $\pi^c(1)=0$ (see Example \ref{ex bubble}
below).  In particular:

\begin{lemma}
\label{lemma pic(1)=0}
$\pi^c(1)=0$ if and only if $q(\theta)\ge0$ for every $\theta\in\Theta_*$.
\end{lemma}

Define
\begin{equation}
\label{K}
\K=\{f\in\Fun(\Omega):\pi(\abs f)<\infty\}
\quad\text{and}\quad
\K_*=\{f\in\K:f_*>-\infty\}
\end{equation}
The set $\K$ plays in this paper the role of the ambient space and it is interesting
to remark that its definition is entirely market based and does not require any 
mathematical structure.

The following is the most important result of this section.

\begin{theorem}
\label{th coherent}
The price functional $q$ is coherent if and only if for each $h\in\B_*$ the set
\begin{equation}
\label{M}
\M
=
\left\{m\in ba_{*,+}:
\K\subset L(m)
\text{ and }
\int fdm\le\pi(f)
\text{ for all }f\in\K_*\right\}
\end{equation}
contains an element $m_h$ such that $\int hdm_h=\pi(h)$. 
\end{theorem}

\begin{proof}
By Lemma \ref{lemma pi}, if $q$ is coherent then the space $\K$ is a vector sublattice 
of $\Fun(\Omega)$ containing $\B_*$ and $\pi$ a $\ge_*$-monotone, positively 
homogeneous and subadditive functional on $\K$. Fix $h\in\B_*$ and consider the 
set $C_h=\{\lambda h:0\le\lambda\le1\}$. By Theorem \ref{th efficient} there is a 
positive linear functional $\beta_h$ on $\K$ vanishing on $\B_*$ and $m_h\in ba_{*,+}$ 
such that $\K\subset L^1(m_h)$ and $\pi(f)\ge\beta_h(f)+\int fdm_h$ for all $f\in\K$ 
and such that $\pi(h)=\beta_h(h)+\int hdm_h=\int hdm_h$. Suppose that 
$g\in\K_*$. Then, $g^-\in\B_*$ and thus
\begin{align*}
\pi(g)\ge\beta_h(g)+\int gdm_h
=
\beta_h(g^+)+\int gdm_h
\ge
\int gdm_h
\end{align*}
so that $m_h\in\M$. Conversely, if $m\in\M$ and $\bar X(\theta)_*\ge0$ then
\eqref{M} implies
\begin{equation*}
q(\theta)
\ge
\pi(\bar X(\theta))
\ge
\int\bar X(\theta)dm
=
\int\bar X(\theta)\sset{\bar X(\theta)>\bar X(\theta)_*-\varepsilon}dm
\ge
[\bar X(\theta)_*-\varepsilon]m(\Omega)
\ge
-\varepsilon m(\Omega)
\end{equation*}
for every $\varepsilon>0$ so that $q$ is coherent.
\end{proof}

We refer to $\M$ as the set of \textit{pricing measures}. It corresponds to the set 
of \textit{equivalent martingale measures} in traditional models%
\footnote{
Although the notion of \textit{equivalence} has no meaning here.
}. 

As in other papers in this field, Theorem \ref{th coherent} asserts that a coherent 
price system is consistent with risk neutral pricing, i.e. with a pricing rule appropriate 
for a market free of imperfections and thus supports the view expressed in the 
microstructure literature that the bid and ask prices are set starting from a 
\textit{consensus price}. However, one should remark that prices may be identified 
with integrals only for strategies in $\Theta_*$ as claims not included in $\mathscr K$ 
may not be integrable at all. This is a consequence of not defining the ambient space 
exogenously. Moreover it should be noted that pricing measures are just finitely 
additive but in addition are defined on all subsets of $\Omega$ rather than a
given algebra $\A$. This last remark is relevant for the definition of market
completeness which may be here given with no reference to an artificial
family of sets.

Given the exclusive emphasis of the literature on \textit{countably additive} 
pricing measures, we characterize next this special property.

\begin{theorem}
\label{th NFL}
Let $\pi^c(1)>0$%
\footnote{
Theorem \ref{th NFL} may be established without assuming $\pi^c(1)>0$
upon replacing $\pi$ with
\begin{align*}
\pi_\varepsilon(b)
=
\sup\left\{\int bdm:m\in\M,\ m(\Omega)\ge\varepsilon\right\}\varepsilon^{-1}
\qquad b\in\B_*
\end{align*}
but the statement would be less clear to interpret. 
} 
and $\A$ be an algebra including $\Neg_*$. 
The following are equivalent:
\begin{enumerate}[(i)]
\item\label{Mc}
there exists $0\ne\mu\in\M$ such that $\mu$ is countably additive in restriction
to $\A$;
\item\label{pi}
there exists $P\in\Prob_*(\A)$ countably additive and such that for any sequence
$\seqn f$ in $\Sim(\A)$, $\limsup_n\pi(f_n)\le0$ implies $\liminf_n\int f_ndP\le0$;
\item\label{pic}
there exists $P\in\Prob_*(\A)$ countably additive and such that for any sequence
$\seqn f$ in $\Sim(\A)_+$, $\lim_n\int f_ndP=0$ implies $\lim_n\pi^c(f_n)=0$.
\end{enumerate}
\end{theorem}

\begin{proof}
\imply{Mc}{pi}.
Choose $0\ne m_0\in\M$ to be countably additive on $\A$ and write $P$ for 
the restriction to $\A$ of $m_0/\norm{m_0}$. Then it is obvious that 
$\int fdP\le\pi(f)/\norm{m_0}$ for each $f\in\Sim(\A)$ so that \iref{pi} holds. 
\imply{pi}{pic}.
Fix $P$ as in \iref{pi} and, assuming that \iref{pic} fails, pick a sequence 
$\seqn h$ in $\Sim(\A)_+$ which converges to $0$ in $L^1(P)$ but such that 
$\inf_n\pi^c(h_n)>\delta>0$. Write $f_n=1-h_n\pi(1)/\pi^c(h_n)$. Then, 
$f_n\in\Sim(\A)$ converges in $L^1(P)$ to $1$ while 
$\pi(f_n)\le\pi(1)+\pi(-h_n)\pi(1)/\pi^c(h_n)=0$, so that \iref{pi} fails.
\imply{pic}{Mc}.
Let now $P$ be as in \iref{pic} and suppose that no $m\in\M$ satisfies $m\ll P$
in restriction to $\A$.
For each $m\in\M$ we may then construct a sequence $\sseqn{F_n(m)}$ in $\A$ 
such that $\lim_nP(F_n(m))=0<\delta(m)\equiv\inf_nm(F_n(m))$. Upon 
choosing $n$ sufficiently large and setting $h_n(m)=\set{F_n(m)}\delta(m)^{-1}$ 
we obtain $\int h_n(m)dP<2^{-n}$ while $\int h_n(m)dm\ge1$. Let 
$\mathscr H_n=\left\{h\in\Sim(\A)_+:\int hdP<2^{-n}\right\}$. Then
\begin{equation*}
\inf_{m\in\M}\sup_{h\in\mathscr H_n}\int hdm\ge1
\end{equation*}
Observe that $\M$ is convex and weak$^*$ compact and that $\mathscr H_n$ 
is convex. By the minimax Theorem of Sion \cite[Corollary 3.3]{sion}, there 
exists then $h_n\in\mathscr H_n$ such that 
$$
\pi^c(h_n)
=
\inf_{m\in\M}\int h_ndm\ge1/2
$$
The sequence $\seqn h$ so obtained contradicts \iref{pic}.
\end{proof}

Theorem \ref{th NFL} contributes to clarifying that the existence of a countable 
additive pricing measure is equivalent to some form of continuity of market prices 
with respect to the $L^1(P)$ topology, an extremely unlikely property in the
absence of an \textit{ad hoc} assumption.

Traditionally, the existence of a countably additive pricing measure is obtained
after imposing the \textit{No-Free-Lunch} condition introduced by Kreps \cite{kreps}
which however requires the choice of  $L^p(P)$ as the ambient space for some 
given probability $P$. We adapt from \cite[Definition 2.1]{jouini kallal} the 
following:

\begin{definition}
\label{def NFL}
Financial markets are said to satisfy the (NFL) condition if there exists $P\in\Prob(\A)$ 
countably additive such that for each sequence $\seqn x$ in $\R$ converging to 
some $x\ge0$ and all sequence $\seqn f$ in $\Sim(\A)$ converging to $f$ in 
$L^1(P)$ and such that $\pi(1)x_n+\pi(f_n)\le0$ one has $\int fdP\le-x$.
\end{definition}

If $P$ is obtained from some $0\ne m\in\M$ which is countably additive in restriction 
to $\A$ and the sequences $\seqn x$ and $\seqn f$ are as in Definition \ref{def NFL}, 
then
\begin{align*}
0\ge\lim_n\{\pi(1)x_n+\norm{m}\int f_ndP\}=\pi(1)x+\norm m\int fdP
\end{align*}
so that $\int fdP\le-(\pi(1)/\norm m)x\le-x$. Conversely, the sequence $\seqn f$ in 
$\Sim(\A)$ constructed to prove the implication \imply{pi}{pic} in Theorem \ref{th NFL} 
is such that $f_n$ converges to $1$ in $L^1(P)$ while $x_n=-\pi(f_n)\ge0$, contradicting
(\textit{NFL}). This proves that

\begin{corollary}
\label{cor NFL}
Under the conditions of Theorem \ref{th NFL}, there exists $0\ne\mu\in\M$ which is 
countably additive in restriction to $\A$ if and only if (NFL) holds.
\end{corollary}

Theorem \ref{th NFL} also provides some insight, suggesting cases in which $\M$ 
may admit no countably additive elements.

\begin{example} 
\label{ex no ca}
Let $\Omega$ be a separable metric space and $\A$ its Borel $\sigma$
algebra. Assume that there exists an increasing net $\neta N$ in $\Neg_*$ 
with $N_\alpha$ open and $\Omega=\bigcup_\alpha N_\alpha$. This is
the case, e.g., if each $\omega\in\Omega$ admits a neighborhood contained 
in $\Neg_*$. Fix $P\in\Prob(\A)$ countably additive. By \cite[Proposition 7.2.2]{bogachev}, 
$1=\lim_\alpha P(N_\alpha)=\lim_k P(N_k)$ for some suitable
sequence $\seq Nk$ from $\neta N$. Set $f_k=\set{N_k^c}$. Then, 
$\lim_k\int f_kdP=0$ while
$$
\pi^c(f_k)
=
\inf_{m\in\M}m(N_k^c)
=
\inf_{m\in\M}m(\Omega)
=
\pi^c(1)
$$
so that condition (\textit{ii}) of Theorem \ref{th NFL} fails. No pricing measure
is then countably additive outside of the special case $\pi^c(1)=0$. Actually, 
decomposing each $m\in\M$ as $m=m^c+m^\perp$, with $m^c$ countably 
additive and $m^\perp$ purely finitely additive (see \cite[III.7.8]{bible}), and 
exploiting the inclusion $\M\subset ba_*$, we conclude that in the case of this
example all pricing measures are purely finitely additive.
\end{example}

The special situation illustrated in Example \ref{ex no ca} highlights that countable 
additivity of the pricing measures may not only fail but actually contrast with 
coherence if the partial order $\ge_*$ is an \textit{a priori} of the model.

\section{Coherent Bubbles}
\label{sec decomposition}
Based on the results of the preceding section, we develop here some decompositions
of coherent price functionals which highlight the role of asset bubbles.

\begin{theorem}
\label{th beta}
The price $q$ is coherent if and only if the set $\M$ of pricing measures 
is the unique non empty, convex, weak$^*$ compact subset of 
$ba_*$ admitting the decomposition
\begin{equation}
\label{beta}
\pi(f)=\beta(f)+\sup_{m\in\M}\int fdm
\qquad\text{for all }
f\in\K
\end{equation}
where $\beta:\K\to\R$ vanishes on $\B_*$.
\end{theorem}

\begin{proof}
Indeed if $q$ is coherent then $\M$ is a non empty, convex and weak$^*$ compact 
subset of $ba$, by Lemma \ref{lemma M}; moreover, if $m\in\M$ and $f\in\K$ then, 
by \eqref{M}, $\dabs{\int f dm}\le\int\abs f dm\le\pi(\abs f)<\infty$ so that \eqref{beta} 
may be regarded as an implicit definition of $\beta$. 
Observe that from \eqref{M} and Lemma \ref{lemma M} we obtain
\begin{equation}
\label{beta bounds phi}
\begin{split}
\inf_{\phi\in\Phi(\pi)}\phi^\perp(f)
	&\le
\pi(f)-\sup_{m\in\M}\int fdm\\
	&=
\beta(f)\\
	&=
\sup_{\phi\in\Phi(\pi)}\phi(f)-\sup_{m\in\M}\int fdm\\
	&\le
\sup_{\phi\in\Phi(\pi)}\phi^\perp(f)+\sup_{\mu\in\M}\int fd\mu-\sup_{m\in\M}\int fdm\\
	&=
\sup_{\phi\in\Phi(\pi)}\phi^\perp(f)
\end{split}
\end{equation}
Given that $\sup_{\phi\in\Phi(\pi)}\phi^\perp(f)%
=0$ for all $f\in\B_*$, as we showed in the proof of Theorem \ref{th coherent}, 
we conclude that $\beta$ vanishes on $\B_*$. This proves existence. To show
uniqueness, suppose that $\bar\beta$ and $\bar\M$ is another pair with the same 
properties of $\beta$ and $\M$ and for which the decomposition \eqref{beta} holds. If 
$\mu\in\bar\M\backslash\M$, then there exists $f\in\B$ such that 
$\sup_{\bar m\in\bar\M}\int fd\bar m\ge\int fd\mu>\sup_{m\in\M}\int fdm$ but 
$\bar\beta(f)=\beta(f)=0$, a contradiction of \eqref{beta}. To show that
\eqref{beta} is sufficient for $q$ to be coherent, let $f\in\K_*$.
Then, $f^-\in\B_*$ and thus $\beta(f)=\beta(f^+)\ge0$ and thus
\begin{align*}
\pi(f)
	\ge
\sup_{m\in\M}\int fdm
	\ge
f_*\sup_{m\in\M}\norm m
	=
f_*\pi(1)
\end{align*}
Therefore, if $\bar X(\theta)\ge_*0$ for some $\theta\in\Theta$ then 
$q(\theta)\ge\pi(\bar X(\theta))\ge0$ and $q$ is coherent.
\end{proof}


For each $m\in\M$ the quantity $\int\bar X(\theta)dm$ is rightfully interpreted 
as the fundamental value of the portfolio $\theta$ \textit{given} $m$. In order 
to overcome the arbitrariness implicit in having a multiplicity of possible pricing 
measures and obtain an unambiguous definition, it is correct to identify the fundamental 
value of $\theta$ with the quantity
\begin{equation}
\label{fundamental}
\sup_{m\in\M}\int\bar X(\theta)dm
\end{equation}
Of course, the supremum of a family of integrals may be represented as the 
Choquet integral with respect to a supermodular capacity having $\M$ as its core%
\footnote{
The use of capacities in finance was introduced by Chateauneuf et al.
\cite[Theorem 1.1]{chateauneuf} precisely with the aim of modeling transaction 
costs. In their paper, however, this representation is an assumption 
(see also \cite{cerreia}). 
}.
Differently from classical asset pricing formulas, the fundamental value is not linear
here, due to transaction costs. It could be interpreted as the maximum price paid 
for $\theta$ in an economy identical with the one considered above but with 
no transaction costs. The main point is not only the multiplicity of pricing 
measures, which would be prevalent even in economies with incomplete 
financial markets, but rather the fact that the intervening expectations do not 
agree on the set of traded payoffs so that the integral appearing in 
\eqref{fundamental} is not invariant with respect to the choice of $m\in\M$. 

In general, deviations of prices from fundamental values are interpreted 
in the literature as evidence of the existence of bubbles. See 
\cite{cox hobson}, \cite{hugonnier} or \cite{lowenstein willard} for 
examples of models dealing with bubbles in continuous time. In so doing, 
however, inefficiency phenomena and the potential contribution of asset bubbles 
to an efficient pricing are mixed together.

Inefficiency is measured by the quantity $q(\theta)-\pi(\bar X(\theta))$. The 
empirical literature typically reports a relatively large number of violations, e.g., 
of the PUT/CALL parity, by which, say, a CALL option may be replaced by a less 
costly synthetic constructed using the corresponding PUT, future and riskless 
asset. Luttmer \cite{luttmer}, takes this mispricing as the sole source of subadditivity. 
For a coherent price system inefficiencies are a consequence of the restrictions 
which prevent investors to exploit them to obtain immediate profits. Empirical 
explanations, such as those invoked by Lamont and Thaler \cite{lamont thaler}, 
draw attention on the fixed costs of trading which impair the arbitrage profits 
emerging from considering prices only. However, even fixed transaction costs 
would play virtually no role if investors were not somehow constrained in their 
ability to either take short positions or in choosing the scale for their investments 
arbitrarily large.

We deduce from \eqref{beta} that, even in the absence of market inefficiencies
and with only two dates, prices may differ from fundamental values by a bubble 
component, $\beta$, interpreted as the price of the tail part of the asset discounted 
payoff. We base this interpretation on the inequality
\begin{align}
\label{bubble}
\dabs{\beta\left(\bar X(\theta)\right)}
\le
\lim_n\left\{\pi\left(\left(\bar X(\theta)^+-n\right)^+\right)
+
\pi\left(\left(\bar X(\theta)^--n\right)^+\right)\right\}
\end{align}
By \eqref{bubble}, $\beta\left(\bar X(\theta)\right)$ is rightfully viewed 
as the component of the price of $\theta$ which only depends on the event 
$\left\{\dabs{\bar X(\theta)}\ge n\right\}$ for all $n\in\N$, i.e. on the 
extreme fluctuations of the portfolio discounted payoff. Observe that
necessarily the price of the \textit{num\'eraire} and of other derivatives 
written on it, such as futures and options, admits no bubbles. \eqref{bubble} 
suggests in addition that, like in other models, bubbles are related 
to the limit of the price of a CALL option as the strike price increases to 
infinity. 
This finding is consistent with similar conclusions linking the existence of 
asset bubbles to some mispricings of options (see \cite{cox hobson} and 
\cite{heston}). Assuming some form of 
monotone continuity of the pricing functional, as in \cite{chateauneuf}, 
excludes the existence of bubbles.

The following example illustrates the economic role of bubbles in a special case.

\begin{example}[Efficient Bubbles]
\label{ex bubble}
Consider a market on which $\bar X(\theta)\ge_*0$ for all $\theta\in\Theta$
and assume that $q(\theta)=X(\theta)^*-X(\theta)_*$. The price function
is clearly subadditive and positively homogeneous. Given that $q(\theta)\ge0$
for all $\theta\in\Theta$ it is coherent too -- although it will be inefficient in general.
From Lemma \ref{lemma pic(1)=0} we know that $\pi^c(1)=0$.
Assume that $X_{0,*}=0$ and $X_0^*=1$ and that, for each $n\in\N$ there 
exists $\theta_n\in\Theta$ with $X(\theta_n)=a+X_02^{-n}$ with $a>0$.
Then, $\bar X(\theta_n)\ge a+2^{-n}$ so that $\bar X(\theta_n)_*\ge a+2^{-n}$. 
On the other hand $X(\theta_n)_*=a$ and $X(\theta_n)^*=a+2^{-n}$. But then
$q_0=1$ while
\begin{equation}
\pi(1)
=
\inf_{\bar X(\theta)_*>0}q(\theta)/\bar X(\theta)_*
\le
\inf_n(1+a2^n)^{-1}
=
0
\end{equation}
Thus the \textit{num\'eraire} is priced inefficiently and, from \eqref{beta}, the only 
possible non null efficient price is a bubble.
\end{example}

\section{Option Pricing}\label{sec option}

In this section we apply our preceding results to option pricing, under the only 
assumption that options are traded anonymously (i.e. that \eqref{anonymity} 
holds) and at non negative prices. $X>0$ will be hereafter the payoff of a 
given underlying and $K(X)$ the set of strike prices (including $k=0$) of all 
CALL options written on it. The ticker of each of these options and the corresponding 
strategy, price and payoff will be indicated by $\alpha_X(k)$, $\theta_X(k)$, 
$q_X(k)$ and $X(k)$ respectively. Define also
\begin{equation}
\TA_X=\{\alpha_X(k):k\in K(X)\}
\quad\text{and}\quad
\Theta_X=\{\theta\in\Theta:\theta\ge0\text{ and }\theta(\alpha)=0
\text{ whenever }\alpha\notin\TA_X\}
\end{equation}

\begin{equation}
\label{piX}
\pi_X(h)
=
\inf\left\{\lambda q(\theta):\lambda\frac{X(\theta)}{X\wedge1}\ge_*h, 
\lambda>0,
\theta\in\Theta_X\right\}
\qquad
h\in\Fun(\Omega)
\end{equation}
and
\begin{equation}
\label{KX}
\K_X=\left\{h\in\Fun(\Omega):\pi_X(\abs h)<\infty\right\}
\end{equation}
Observe that the restriction of $q$ to $\Theta_X$ is coherent, given our assumption 
of non negative prices. On the other hand, the change of 
\textit{num\'eraire} implicit in \eqref{piX} entails a different concept of efficiency. 
In particular we shall say that options are priced efficiently if 
\begin{equation}
\label{X efficiency}
q_X(k)=\pi_X\left(\frac{X(k)}{X\wedge1}\right)
\qquad
k\in K(X)
\end{equation}

The criterion adopted in \eqref{X efficiency} is indeed quite weak as, for example, 
it does not involve PUT options nor Futures or short positions. This is desirable since 
the larger the set of derivatives involved the more likely is it that efficiency may fail. 
For example, the PUT/CALL parity is well known to generate a large number of 
violations as well as the lower bound for CALL options%
\footnote{Cerreia et al. \cite{cerreia} construct a 
pricing model for markets which are assumed to satisfy the PUT/CALL parity.}.

Define the set%
\footnote{The limit appearing in \eqref{Gamma} exists by convexity.}
\begin{equation}
\label{Gamma}
\Gamma
=
\left\{f\in\Fun(\R_+):
f\ge0=f(0),\ f\text{ convex},\ \lim_{n\to\infty}f(n)/n<\infty\right\}
\end{equation}

To start discussing the issue of options efficiency, denote by $J(X)\subset K(X)$ 
the subset of strike prices $k$ possessing the following property
\begin{equation}
\label{butterfly}
q_X(k)
\le
\inf\left\{aq_X(k_1)+(1-a)q_X(k_2):
k_1,k_2\in K(X),0\le a\le1,\ ak_1+(1-a)k_2\le k\right\}
\end{equation}
i.e. which satisfy the \textit{butterfly spread} condition. One should remark that in the 
present setting this is not an arbitrage restriction. 

For reasons of technical convenience, in the rest of this section we shall adopt the 
following

\begin{assumption}
\label{ass ATM}
$X^*<\infty$.
\end{assumption}

We turn now to the issue of derivatives hedging. 

\begin{theorem}
\label{th hedging}
Assume that $\left(X\sset{X\le j}\right)^*=j<X^*$ for each $j\in J(X)$.
For each $g\in\Gamma$ there exists $\theta_X(g)\in\Theta_X$ such that 
$q(\theta_X(g))=\pi_X(g(X)/X\wedge1)$. Moreover: (i) if $g_1,g_2\in\Gamma$
then $\theta_X(g_1+g_2)=\theta_X(g_1)+\theta_X(g_2)$, (ii) if $j\in J(X)$ and
$g(x)=(x-j)^+$ then $\theta_X(g)=\theta_X(j)$.
\end{theorem}

\begin{proof}
The existence claim is proved in Lemma \ref{lemma g} in the Appendix where 
the explicit composition of $\theta_X(g)$ is described, see \eqref{thetaX(g)}. 
From it we deduce (\textit{i}). (\textit{ii}) follows upon setting $g(x)=(x-j)^+$ in
\eqref{w recursion}.
\end{proof}

We deduce from Theorem \ref{th hedging} that an option is priced efficiently
if and only if its strike price is included in $J(X)$. 

The following is the most important result of the paper.

\begin{theorem}
\label{th option}
Assume that $(X\sset{X\le x})^*=x\wedge X^*$ when $x\ge0$. Let
$G=\{g_t:t\in\R_+\}\subset\Gamma$ satisfy 
\begin{equation}
\label{G}
g_t\le ag_{t_1}+(1-a)g_{t_2}
\quad\text{whenever}\quad
0\le a\le1
\text{ and } 
t\ge at_1+(1-a)t_2
\end{equation}
There exist
$\beta^G(X)\ge0$ and $\nu_X^G\in ca(\mathscr B(\R_+))_+$ such that
\begin{equation}
\label{option}
\pi_X(g_t(X)/X\wedge1)
=
\beta^G(X)+\int_t^\infty\nu^G_X(x>z)dz
\qquad\text{for all } t\ge0
\end{equation}
\end{theorem}

\begin{proof}
The function $t\to q_X^G(t)=\pi_X(g_t(X)/X\wedge1):\R_+\to\R_+$ is clearly 
decreasing and convex and thus satisfies \eqref{butterfly}. Replace the original 
option market with one in which all strikes $0\le t\le X^*$ are traded at the 
fictitious prices $q^G_X(t)$ and define $\pi_X^G$ exactly as in \eqref{piX} after 
such replacement. By construction, all option prices are efficient, i.e. $J^G(X)=[0,X^*]$.
Moreover, if $a_1,\ldots,a_N\ge0$ and $t_1,\ldots,t_N\in[0,X^*]$ then
\begin{align*}
\pi_X^G\left(\frac{\sum_{n=1}^Na_nX(t_n)}{X\wedge 1}\right)
	=
\sum_{n=1}^Na_n\pi_X^G\left(\frac{X(t_n)}{X\wedge 1}\right)
	=
\sum_{n=1}^Na_nq_X^G(t_n)
\end{align*}
This follows clearly from Theorem \ref{th hedging} if one tries to hedge the payoff 
$\sum_{n=1}^Na_nX(t_n)/X\wedge 1$ with a finite set of options whose strikes 
include $t_1,\ldots,t_N$. Write
\begin{equation*}
\mathscr K_X^G
	=
\{f\in\Fun(\Omega):\pi_X^G(\abs f)<\infty\}
\quad\text{and}\quad
\mathscr C_X^G
	=	
\left\{\frac{\sum_{0\le t\le X^*}a(t)X(t)}{X\wedge 1}:a\in\Fun_0([0,X^*])_+\right\}
\end{equation*}
endowed with the partial order $\ge_*$ and observe that $\B_*\subset\mathscr K_X^G$.

By Theorem \ref{th efficient} we obtain a $\ge_*$ positive, linear functional 
$\phi^G_X:\K^G_X\to\R$ such that $\phi^G_X\le\pi^G_X$ and that $\phi^G_X=\pi^G_X$ 
in restriction to $\mathscr C^G_X$. Define
\begin{equation*}
F^G(t)
	=
\phi^G_X\left(\frac{X(t)}{X\wedge1}\right)
\qquad 
t\ge0
\end{equation*}
Of course, $F^G(t)=q^G_X(t)$; in addition, it is decreasing and 
convex. By a standard result on convex functions, we may write
\begin{equation}
\label{convex representation}
F^G(t_2)=F^G(t_1)+\int_{t_1}^{t_2}f^G(t)dt
\qquad 0<t_1<t_2
\end{equation}
where, for definiteness, we take $f^G(t)$ to be the right derivative of $F^G$ for 
$t\in\R_+$. Suppose that $\{u\ge X>t\}\in\Neg_*$ for some $0\le t<u$ and fix 
$0<h\le(u-t)/2$. There is then a negligible set outside of which each of the options 
with strike prices $t,t+h,u-h,u$ expires in the money if and only if all the others do. 
In other words 
\begin{equation*}
\frac{X(t)}{X\wedge1}+\frac{X(u)}{X\wedge1}
=_*
\frac{X(t+h)}{X\wedge1}+\frac{X(u-h)}{X\wedge1} 
\end{equation*}
from which it follows
\begin{align*}
F^G(t)+F^G(u)
	&=
\phi^G_X\left(\frac{X(t)+X(u)}{X\wedge1}\right)
	=
\phi^G_X\left(\frac{X(t+h)+X(u-h)}{X\wedge1}\right)
	=
F^G(t+h)+F^G(u-h)
\end{align*}
Thus,
\begin{align*}
\frac{F^G(u)-F^G(u-h)}{h}=\frac{F^G(t+h)-F^G(t)}{h}
\end{align*}
i.e. the left derivative of $F^G$ at $u$ and the right derivative of $F^G$ at $t$ 
coincide. There exists then a set $D\subset\R_+$ with $\R_+\backslash D$ 
at most countable and such that 
$\{X>u\}\bigtriangleup N_1=\{X>t\}\bigtriangleup N_2$ for $t,u\in D$ and 
$N_1,N_2\in\Neg_*$ imply $f(t)=f(u)$. It is therefore possible to define a positive 
set function $\lambda^G_0$ on the collection $\A_0(X)$ of subsets of $\Omega$ 
formed by $\Omega$, $\emp$ and all sets of the form $\{X>t\}\bigtriangleup N$ 
with $t\in\R_+$ and $N\in\Neg_*$ implicitly by letting
$\lambda^G_0(\Omega)=-f^G(0)$, $\lambda^G_0(\emp)=0$ and
\begin{equation}
\label{lambda}
\lambda^G_0(\{X>t\}\bigtriangleup N)
	=
\sup_{\{u\in D:u\ge t\}}-f^G(u)
\qquad t\in\R_+,\ N\in\Neg_*
\end{equation}
To see that this definition is well taken, observe that, if there is $t_\infty\in D$ such that 
$\{X>t_\infty\}=\emp$ then $F^G(t_\infty+h)=F^G(t_\infty)$ and so $f^G(t_\infty)=0$. 
Likewise, if $\{X>t_0\}=\Omega$ for some $t_0\in D$, then $\{X>t_0\}=\{X>0\}$ so 
that, as seen above, $f^G(t_0)$ coincides with $f^G(0)$. If either $t_0$ or $t_\infty$ do not 
exist, then one can choose the corresponding value of $\lambda^G_0$ arbitrarily. Since 
the elements of $\A_0(X)$ are linearly ordered by inclusion it follows that if
$A_i=\{X>t_i\}\bigtriangleup N_i$ and $N_i\in\Neg_*$ for $i=1,2$ with $t_1\ge t_2$
then
\begin{align*}
\lambda^G_0(A_1)+\lambda^G_0(A_2)
&=
\lambda^G_0(X>t_1)+\lambda^G_0(X>t_2)\\
&=
\lambda^G_0(\{X>t_1\}\cap\{X>t_2\})+\lambda^G_0(\{X>t_1\}\cup\{X>t_2\})\\
&=
\lambda^G_0(A_1\cap A_2)+\lambda^G_0(A_1\cup A_2)
\end{align*}
as 
$\{X>t_1\}\bigtriangleup (A_1\cap A_2),\{X>t_2\}\bigtriangleup (A_1\cup A_2)%
	\subset 
N_1\cup N_2
	\in
\Neg_*$. 
It follows from \cite[Theorems 3.1.6 and 3.2.10]{rao}, that there exists a unique 
extension $\lambda^G\in ba(\A(X))_+$ of $\lambda^G_0$ to the algebra $\A(X)$ 
generated by $\A_0(X)$ and thus such that $\lambda^G(N)=0$ when $N\in\Neg_*$. 
Let
\begin{equation}
\label{option price lim}
\beta^G_X=\lim_{k\to\infty}q_X^G(k)
\end{equation}
Then we obtain from \eqref{convex representation}
\begin{align*}
\beta^G_X=F^G(k)-\int_k^\infty\lambda^G(X>t)dt
\qquad k\ge0
\end{align*}
To eventually get \eqref{option}, write 
\begin{equation*}
\A=\left\{A\subset\R_+:X^{-1}(A)\in\A(X)\right\}
\end{equation*}
It is clear that $\A$ is an algebra containing the algebra $\A(\R_+)$ generated by 
the left open intervals of $\R_+$. Define then $\lambda^G_X\in ba(\A(\R_+))$ by 
letting $\lambda^G_X(A)=\lambda^G(X\in A)$ and observe from \eqref{convex 
representation} that $\int_{\R_+}\lambda^G_X(x>t)dt=-\int_{\R_+} f(t)dt\le F^G(0)$ 
so that $\lim_t\lambda^G_X((t,\infty))=0$. Exploiting standard rules of the Lebesgue 
integral and integration by parts we obtain
\begin{equation*}
\int_{\R_+}\lambda^G(X>t)dt
	=
\int_{\R_+}\lambda^G_X(x>t)dt
	=
\int_{\R_+}xd\lambda^G_X(x)
	=
\int xd\lambda^G_X(x)
\end{equation*}
as $\lambda^G_X(x<-t)=0$ for all $t\ge0$. It follows from \cite[Lemma 2, p. 191]%
{dubins savage} that, uniquely associated with $\lambda^G_X$ is its \textit{conventional 
companion} $\nu^G_X\in ca(\A(\R_+))_+$ with the property that
\begin{equation}
\label{conventional}
\int h(x)d\lambda^G_X=\int h(x)d\nu^G_X
\end{equation}
for any continuous function $h:\R\to\R$ for which either integral is well defined. The 
extension from $\A(\R_+)$ to the generated $\sigma$ algebra $\mathscr B(\R_+)$ is 
standard. Thus the representation \eqref{option} is implicit in $\Theta^G_X$ being 
priced efficiently. Suppose that 
$\hat\beta^G_X\ge0$ and $\hat\nu^G_X\in ca(\mathscr B(\R_+))_+$ 
is another pair for which the representation
\eqref{option} holds. Then, 
$\beta^G_X-\hat\beta^G_X=\int_k^\infty[\hat\nu^G_X(x>t)-\nu^G_X(x>t)]dt$ for 
all $k\ge0$ which implies $\beta^G_X=\hat\beta^G_X$. 
\end{proof}

Observe that, by standard rules,
\begin{equation}
\int_t^\infty\nu^G_X(x>z)dz=\int(x-t)^+d\nu_X^G(x)
\end{equation}
Thus \eqref{option}, represents the price of the $G$ derivatives as the sum of 
a bubble part and the fundamental value. By \eqref{option price lim} the term 
$\beta^G(X)$ represents the (fictitious) option price as the strike approaches 
infinity and contributes to explaining the overpricing of deeply out of the money 
CALL's often documented empirically in some form of the smile effect.

The most important implications of Theorem \ref{th option} regard the empirical 
analysis of option markets. In this perspective one should start noting that the CALL 
function $q_X^G$, although not a quoted price, is entirely market based and may 
be computed explicitly, once the collection $G$ has been chosen. Statistical estimation 
of the CALL function, to the contrary, follows from some optimal statistical criterion
and does not guarantee a direct market interpretation. In principle  one could make
the choice of the collection $G$ sample based and study whether the estimate
$q_X^G(t)$ possesses reasonable statistical properties. This implicitly suggests 
a new non parametric empirical strategy.

A second fact arising from \eqref{option} is the representation of option prices via a 
\textit{countably additive} probability  $\nu_X^G$ implicit in option prices. Although 
$\nu_X^G$ will generally depend on $G$, \eqref{option} makes it possible, even in 
a model with minimal mathematical structure as the one developed here, to run the 
classical exercise of Breeden and Litzenberger \cite{breeden litzenberger} and Banz 
and Miller \cite{banz miller} by computing
\begin{equation}
\label{breeden}
\nu_X^G(x>t)=-\left.\dDer{q_X^G(k)}{k}\right\vert_{k=t}
\qquad\text{for all } 
t\ge0
\end{equation}
Remark that in the model of Black and Scholes \eqref{breeden} translates into
the classical formula
\begin{equation}
\label{BS}
\nu^{BS}_X(x>k)=e^{-rT}\Phi(d_2)
\quad\text{with}\quad
d_2=\frac{\ln(S_0/k)+(r-\frac12\sigma^2)T}{\sigma\sqrt T}
\end{equation}
so that $\norm{\nu^{BS}_X}=\exp(-rT)$. In more general traditional models, 
the CALL lower bound, $q_X\ge q_X(k)\ge q_X-kq_0$ implies the inequality 
$\norm{\nu_X}\le q_0\le1$. However, this inequality cannot be deduced from
arbitrage arguments when the \textit{num\'eraire} asset is not free of risk. In 
particular, we do not have any \textit{a priori} bound to impose on the norm of 
$\nu_X$. If, however, we assume in addition $X\ge_*\eta$, we then obtain
for $G=\{(x-t)^+:t\in\R_+\}$
\begin{equation*}
\norm{\nu^G_X}
	=
\nu^G_X(x>0)
	=
\lim_{k\to0}\frac{q^G_X(0)-q^G_X(k)}{k}
	\le
\pi_X(1)
\le
\frac{q_X(0)}{\eta}
\end{equation*}

\appendix

\section{Auxiliary Results}

Let us start with two general results.

\begin{theorem}
\label{th cone}
Let $\succeq$ be a partial order on $\Fun(\Omega)$ satisfying Assumption
\ref{ass preorder} and $\Gamma\subset\Fun(\Omega)$ a convex cone.
There is no $f\in\Gamma$ with $f\succeq1$ if and only if there exists 
$m\in\Prob_\succeq=\{\mu\in\Prob:\mu(N)=0\text{ when }0\succeq\set N\}$ 
such that
\begin{equation}
\Gamma_\succeq
	=
\{f\in\Gamma:f\succeq a\text{ for some }a\in\R\}
	\subset
L^1(m)
\quad\text{and}\quad
\sup_{f\in\Gamma_\succeq}\int fdm\le0
\end{equation}
\end{theorem}

\begin{proof}
Define the following collection
\begin{equation*}
\Gamma_1
	=
\left\{g\in\Fun(\Omega):
g^-\in\B,\ 
\sum_{i=1}^If_i\set{N_i^c}\succeq g
\text{ for some }
f_1,\ldots,f_I\in\Gamma_\succeq,\ 
N_1,\ldots,N_I\in\Neg_\succeq
\right\}
\end{equation*}
Observe that $\Gamma_1$ is a convex cone of lower bounded functions on 
$\Omega$. Suppose that $\Gamma_1$ contains a sure win, i.e. an element 
$g\ge1$. Then there exist $f_1,\ldots,f_I\in\Gamma_\succeq$ and 
$N_1,\ldots,N_I\in\Neg_\succeq$ such that $\sum_{i=1}^If_i\set{N_i^c}\succeq g$. 
However, remarking that $N=\bigcup_iN_n\in\Neg_\succeq$ (so that
$\alpha\set{N}\succeq0$ for any $\alpha\in\R$)
\begin{align*}
\sum_{i=1}^If_i
	\succeq
1+\sum_if_i\set{N_i}
	\succeq
1+\min_i(f_{i,\succeq})\set{N}
	\succeq
1
\end{align*}
contradicting $\sum_{i=1}^If_i\in\Gamma$. By \cite[Proposition 1]{JMAA}, 
there exists $m\in\Prob$ such that 
\begin{equation*}
\Gamma_1\subset L^1(m)
\quad\text{and}\quad
\sup_{g\in\Gamma_1}\int gdm\le0
\end{equation*}
$N\in\Neg_\succeq$ implies $0\succeq\set N$ so that $\set N\in\Gamma_1$, 
since $0\in\Gamma_\succeq$. Therefore $m(N)\le0$ and $m\in\Prob_\succeq$. 
On the other hand $f\in\Gamma_\succeq$ implies 
$f_n=f\sset{f\ge f_\succeq-2^{-n}}\in\Gamma_1$ so that 
$\int\abs fdm=\int\abs{f_n}dm<\infty$ and $\int fdm=\int f_ndm\le0$. For the 
converse, suppose that $f\in\Gamma$ and $f\succeq1$. Then 
$\{f\le1/2\}\in\Neg_\succeq$ and $f\in\Gamma_\succeq$ so that 
$0\ge\int fdm=\int_{\{f>1/2\}}fdm\ge1/2$, a contradiction.
\end{proof}

\begin{lemma}
\label{lemma phi}
Let $\La\subset\Fun(\Omega)$ be a vector lattice containing $\B$. Each
positive linear functional $\phi$ on $\La$ admits the decomposition
\begin{equation}
\label{phi decomposition}
\phi(f)=\phi^\perp(f)+\int fdm_\phi
\qquad
f\in\La
\end{equation}
where $m_\phi\in ba$ and $\phi^\perp$ is a positive linear functional
vanishing on $\B$.
\end{lemma}

\begin{proof}
See \cite[Theorem 1]{JMAA}.
\end{proof}

\begin{theorem}
\label{th efficient}
Let $\La\subset\Fun(\Omega)$ be a vector lattice containing $\B_*$, 
$C\subset\La$ a convex set containing the origin and $\gamma:\Fun(\La)$ 
be $\ge_*$-monotone, subadditive and positively homogeneous. Then
\begin{equation}
\label{linear}
\gamma\left(\sum_{n=1}^Nf_n\right)
=
\sum_{n=1}^N\gamma(f_n)
\qquad
f_1,\ldots,f_N\in C
\end{equation}
if and only if there exist (i) a positive linear functional $\beta$ on $\La$ 
vanishing on $\B_*$ and (ii) $m\in ba_{*,+}$ such that $\La\subset L^1(m)$
\begin{equation}
\label{local}
\gamma(h)\ge\beta(h)+\int hdm
\quad\text{and}\quad
\gamma(f)=\beta(f)+\int fdm
\qquad\text{for all } 
h\in\La,\ f\in C
\end{equation}
\end{theorem}

\begin{proof}
\eqref{linear} holds on $C$ if and only if it holds over the whole
convex cone generated by $C$, by positive homogeneity and the inclusion
$0\in C$. Let $f_1,\ldots,f_N,g_1,\ldots,g_K\in C$ and 
$\lambda_1,\ldots,\lambda_N,\alpha_1,\ldots,\alpha_K\in\R$ be such that
$\sum_{k=1}^K\alpha_kg_k=_*\sum_{n=1}^N\lambda_nf_n$. Then,
$\sum_{k=1}^K\alpha^+_kg_k+\sum_{n=1}^N\lambda^-_nf_n
=_*
\sum_{k=1}^K\alpha^-_kg_k+\sum_{n=1}^N\lambda^+_nf_n$.
By \eqref{linear} and $\ge_*$ monotonicity
\begin{align*}
\sum_{k=1}^K\alpha^+_k\gamma(g_k)+\sum_{n=1}^N\lambda^-_n\gamma(f_n)
&=
\gamma\left(\sum_{k=1}^K\alpha^+_kg_k+\sum_{n=1}^N\lambda^-_nf_n\right)\\
&=
\gamma\left(\sum_{k=1}^K\alpha^-_kg_k+\sum_{n=1}^N\lambda^+_nf_n\right)\\
&=
\sum_{k=1}^K\alpha^-_k\gamma(g_k)+\sum_{n=1}^N\lambda^+_n\i(f_n)
\end{align*}
i.e. $\sum_{k=1}^K\alpha_k\gamma(g_k)=\sum_{n=1}^N\lambda_n\gamma(f_n)$.
Thus the quantity
\begin{equation*}
\phi_0\left(\sum_{n=1}^N\lambda_nf_n\right)
=
\sum_{n=1}^N\lambda_n\gamma(f_n)
\qquad
f_1,\ldots,f_N\in C,\ 
\lambda_1,\ldots,\lambda_N\in\R
\end{equation*}
implicitly defines a linear functional on the linear span $\mathrm{Lin}(C)$ of $C$. It 
is easy to conclude from \eqref{linear} and subadditivity that
\begin{align*}
\phi_0\left(\sum_{n=1}^N\lambda_nf_n\right)
&=
\phi_0\left(\sum_{n=1}^N\lambda_n^+f_n\right)-%
\phi_0\left(\sum_{n=1}^N\lambda_n^-f_n\right)\\
&=
\gamma\left(\sum_{n=1}^N\lambda_n^+f_n\right)-%
\gamma\left(\sum_{n=1}^N\lambda_n^-f_n\right)\\
&\le
\gamma\left(\sum_{n=1}^N\lambda_nf_n\right)
\end{align*}
and thus that $\phi_0\le\gamma$ on  $\mathrm{Lin}(C)$. By Hahn Banach we 
may thus find an extension $\phi$ of $\phi_0$ to the whole of $\La$ such that 
$\phi\le\pi$. Given that $\gamma$ is $\ge_*$-monotone and positive homogeneous 
we conclude that $\phi$ is positive and, by Lemma \ref{lemma phi}, that it admits
the decomposition \eqref{phi decomposition}. Write $\beta=\phi^\perp$ and 
$m=m_{\phi}$. If $N\in\Neg_*$, then $\set N=_*0$ so that 
$0=\phi(\set N)=\beta(\set N)+m(N)=m(N)$ i.e. 
$m\in ba_{*,+}$. Likewise, if $g\in\La$ then 
$
0
\le
\beta(\abs g\set N)
=
\phi (\abs g\set N)
\le
\gamma(\abs g\set N)
\le
0
$ 
so that $\beta$ vanishes on $\B_*$, as claimed. The converse is obvious.
\end{proof}

\section{Proofs}

\begin{proof}[\textbf{Proof of Lemma \ref{lemma subadditive}}]
In any lattice $X$ the operation $x\to x^-$ is subadditive, that is $(x+y)^-\le x^-+y^-$. 
Thus, if  $\X=\{X(\alpha):\alpha\in\TA\}$ and $f,g\in\Fun_0(\X)$
\begin{align*}
q(f+g)
	&=
\sum_{X\in\X}\left\{[(f+g)(X)^+]a(X)-[(f+g)(X)^-]b(X)\right\}\\
	&=
\sum_{X\in\X}\left\{(f+g)(X)a(X)+(f+g)(X)^-(a(X)-b(X))\right\}\\
	&\le
\sum_{X\in\X}\left\{(f+g)(X)a(X)+(f(X)^-+g(X)^-)(a(X)-b(X))\right\}\\
	&=
q(f)+q(g)
\end{align*}
Positive homogeneity is clear. Suppose now that $f,g\in\Fun_0(X)$ satisfy $fg\ge0$ 
that is $f(X)$ and $g(X)$ have the same sign for all $X\in\X$. It is then obvious that 
$(f+g)(X)^-=f(X)^-+g(X)^-$ from which the claim follows.
\end{proof}

\begin{lemma}
\label{lemma pi}
The functional $\pi:\Fun(\Omega)\to\RR$ is $\ge_*$-monotone, positively 
homogeneous and satisfies 
\begin{align}
\label{pi properties}
\pi(\bar X(\theta))\le q(\theta)
\qquad\theta\in\Theta
\qquad\text{and}\qquad
\pi(f+g)\le\pi(f)+\pi(g)
\end{align}
for all $f,g\in\Fun(\Omega)$ for which the sum $\pi(f)+\pi(g)$ is defined. Moreover,
the following properties are equivalent: (i) $q$ is coherent, (ii) $\pi(0)=0$, 
(iii) $\pi^c(1)\le q_0$ and (iv)
\begin{equation}
\label{pi(f+b)}
\abs{\pi(b)}<\infty
\quad\text{and}\quad
\pi(f)+\pi(b^*)\ge\pi(f+b)\ge\pi(f)+\pi^c(b_*)
\qquad\text{for all } f\in\Fun(\Omega),\ b\in\B_*
\end{equation}
\end{lemma}

\begin{proof}
Monotonicity, positive homogeneity and the first part of \eqref{pi properties} are 
obvious properties of $\pi$. Assume that $f,g\in\Fun(\Omega)$ are such that 
$\pi(f)+\pi(g)$ is a well defined element of $\RR$. Thus if, say, $\pi(f)=\infty$ 
then $\pi(f)+\pi(g)=\infty$ and the second part of \eqref{pi properties} is 
obvious. If, alternatively, $\pi(f),\pi(g)<\infty$, then there exist 
$\lambda_f,\lambda_g\ge0$ and $\theta_f,\theta_g\in\Theta$ such that 
$\lambda_f\bar X(\theta_f)\ge_*f$ and $\lambda_g\bar X(\theta_g)\ge_* g$ 
so that $\lambda(\bar X(\theta_f')+\bar X(\theta_g'))\ge_*f+g$, with
$\lambda=\lambda_f+\lambda_g$ and $\theta_f'=\theta_f\lambda_f/\lambda$
and $\theta_g'=\theta_g\lambda_g/\lambda$ (with the convention $0/0=0$). 
Given that, by Assumption \ref{ass Theta}, $\theta=\theta_f'+\theta_g'\in\Theta$
we conclude that 
\begin{align*}
\pi(f+g)
\le
\lambda q(\theta)\le\lambda(q(\theta'_f)+q(\theta'_g))%
=
\lambda_fq(\theta_f)+\lambda_gq(\theta_g)
\end{align*}
and, the inequality being true for all $\lambda_f,\lambda_f$ and $\theta_f,\theta_g$ 
as above, the second half of \eqref{pi properties} follows. \eqref{pi properties} also
implies $\pi(0)\le0$. It is then clear that (\textit{ii}) is equivalent to (\textit{i}). If
$\theta\in\Theta$ and $\lambda\ge0$ are such that $\lambda\bar X(\theta)\ge_*-1$ 
then $(1+\lambda)\bar X\left(\frac{\lambda\theta+\delta_0}{1+\lambda}\right)\ge_*0$ 
so that 
$$
\pi(0)
\le
(1+\lambda)q\left(\frac{\lambda\theta+\delta_0}{1+\lambda}\right)
\le
\lambda q(\theta)+q_0
$$
We thus conclude that $q_0\ge\pi(0)+\pi^c(1)$ and so that $\pi(0)=0$ implies 
$\pi^c(1)\le q_0$. If $b\in\B_*$, then \eqref{pi properties} implies 
$\abs b^*\pi(-1)\le\pi(b)\le\abs b^*\pi(1)$ so that from (\textit{iii}) we deduce 
$\abs{\pi(b)}<\infty$. But then the sums $\pi(f)+\pi(b)$ and $\pi(f+b)+\pi(-b)$ are 
well defined for each $f\in\Fun(\Omega)$ and the second half of \eqref{pi(f+b)} 
follows from \eqref{pi properties}. Conversely, by \eqref{pi(f+b)} we conclude that 
$\pi(0)=n\pi(0)$ for each $n\in\N$ and $\pi(0)\in\R$ so that $\pi(0)=0$.
\end{proof}

\begin{proof}[\textbf{Proof of Lemma \ref{lemma pic(1)=0}}]
If $\pi^c(1)=0$ and $\bar X(\theta)\ge_*0$ then $q(\theta)\ge\pi(-1)$. If
$\bar X(\theta)_*<0$ then $\bar X(\theta)/\abs{\bar X(\theta)_*}\ge_*-1$
and so $q(\theta)/\abs{\bar X(\theta)_*}\ge0$. Conversely, 
$\lambda\bar X(\theta)\ge_*-1$ and $\lambda>0$ imply $\theta\in\Theta_*$
and thus $q(\theta)\ge0$ so that $\pi(-1)\ge0$.
\end{proof}

Denote by 
\begin{equation}
\label{Phi(pi)}
\Phi(\pi)
=
\{\phi\in\Fun(\K): \phi\text{ positive, linear and such that }\phi\le\pi\}
\end{equation}
Adopting the notation of Lemma \ref{lemma phi} we can also write
\begin{equation}
\label{M(pi)}
\M(\pi)=\left\{m_\phi:\phi\in\Phi(\pi)\right\}
\qquad\text{and}\qquad
\Phi^\perp(\pi)=\left\{\phi^\perp:\phi\in\Phi(\pi)\right\}
\end{equation}

\begin{lemma}
\label{lemma M}
If $q$ is coherent then the set $\M$ defined in \eqref{M} is non empty, convex 
and weak$^*$ compact subset of $ba_+$. Moreover, $\M=\M(\pi)$ (see \eqref{M(pi)}).
\end{lemma}

\begin{proof}
If $q$ is coherent, $\M$ is non empty by Theorem \ref{th coherent}. By
\eqref{M}, $\M(\pi)\subset\M$. 
Thus, we only need to prove that $\M$ is closed in the weak$^*$ topology of $ba$ 
and that $\M\subset\M(\pi)$. Let $m_0$ be an element of the closure of $\M$ and $f\in\K$. Then 
$m_0\in ba_{*,+}$ and 
\begin{align*}
\int(\abs f\wedge n)dm_0
\le
\sup_{m\in\M}\int(\abs f\wedge n)dm
\le
\sup_{m\in\M}\int\abs fdm
\le
\pi(\abs f)
\end{align*}
so that the sequence $\sseqn{\abs f\wedge n}$ is Cauchy in $L^1(m_0)$.
Moreover, for all $c>0$%
\footnote{By $v(m)$ we total variation of $m$ as defined in \cite[III.1.9]{bible}.} 
\begin{align*}
v(m_0)(\abs f>c+\abs f\wedge n)
\le
v(m_0)(\abs f>c+n)
\le
\frac{1}{c+n}\int[\abs f\wedge(c+n)]dm_0
\le
\frac{\pi(\abs f)}{c+n}
\end{align*}
which proves that $\abs f\wedge n$ converges to $\abs f$ in $L^1(m_0)$ and so that 
$f\in L^1(m_0)$ \cite[III.3.6]{bible} . Moreover, if $f\in\K_*$ 
\begin{align*}
\int fdm_0
=
\lim_n\int (f\wedge n)dm_0
\le
\sup_{m\in\M}\int fdm
\le
\pi(f)
\end{align*}
which proves that $m_0\in\M$. Observe that by Tychonoff Theorem \cite[I.8.5]{bible}, 
the set $\Phi(\pi)$ is compact in the topology induced on it by $\K$. Let 
$\net m\gamma\Gamma$ be a net in $\M(\pi)$ converging to $m\in\M$ in 
the weak $^*$ topology of $ba$. For each $\gamma\in\Gamma$ there exists 
$\phi_\gamma\in\Phi(\pi)$ such that $m_\gamma=m_{\phi_\gamma}$. By moving to a subnet 
if necessary we obtain that the net $\net\phi\gamma\Gamma$ converges in the 
topology induced by $\K$ to some limit $\phi\in\Phi(\pi)$. Denote by $m_\phi$ the 
part of $\phi$ representable as an integral, as in \eqref{phi decomposition}. If 
$b\in\B$ we have, by the inclusion $\B\subset\K$ that follows from $q$ being 
coherent,
\begin{align*}
\int bdm_\phi
=
\phi(b)
=
\lim_\gamma\phi_\gamma(b)
=
\lim_\gamma\int bdm_\gamma
=
\int bdm
\end{align*}
so that $m=m_\phi\in\M(\pi)$.
\end{proof}

\begin{corollary}
\label{cor beta}
The functional $\beta$ defined in \eqref{beta} is positive and satisfies
\begin{equation}
\label{beta bounds pi}
-\lim_n\{\pi(f)-\pi(f\vee-n)\}\le\beta(f)\le\lim_n\{\pi(f)-\pi(f\wedge n)\}
\qquad f\in\K
\end{equation}
and
\begin{equation}
\label{beta K*}
\beta(f)=\lim_n\{\pi(f)-\pi(f\wedge n)\}
\qquad f\in\K_*
\end{equation}
\end{corollary}

\begin{proof}
Positivity of $\beta$ follows from \eqref{beta bounds phi} 
and the fact that $\phi^\perp(f)\ge0$ for each $f\in\K_*$ and $\phi\in\Phi(\pi)$. 
Moreover,
\begin{align*}
\sup_{m\in\M}\int fdm
=
\sup_{m\in\M}\lim_n\int(f\vee-n)m
\le
\lim_n\sup_{m\in\M}\int(f\vee-n)m
\le
\lim_n\pi(f\vee-n)
\end{align*}
Likewise, given that $\phi^\perp(f\wedge n)\le0$ for all $f\in\K$ and $\phi\in\Phi(\pi)$
by Lemma \ref{lemma M}
\begin{align*}
\sup_{m\in\M}\int fdm
=
\lim_n\sup_{m\in\M}\int(f\wedge n)dm
=
\lim_n\sup_{\phi\in\Phi}\{\phi(f\wedge n)-\phi^\perp(f\wedge n)\}
\ge
\lim_n\pi(f\wedge n)
\end{align*}
and \eqref{beta bounds pi} is proved. \eqref{beta K*} follows from
\begin{align*}
\beta(f)
=
\pi(f)-\sup_{m\in\M}\int fdm
=
\pi(f)-\sup_{m\in\M}\lim_n\int (f\wedge n)dm
=
\pi(f)-\lim_n\sup_{m\in\M}\int (f\wedge n)dm
\end{align*}
and the fact that $\sup_{m\in\M}\int (f\wedge n)dm=\pi(f\wedge n)$ whenever 
$f\in\K_*$.
\end{proof}

In the next results write $J(X)=\{0=j_0<j_1<\ldots<j_I\}$ and $j_{I+1}=X^*$.

\begin{lemma}
\label{lemma un}
Let Assumption \ref{ass ATM} hold.
Let $g\in\Gamma$, $F(X)=\sum_{i=1}^I\alpha_iX(j_i)$. Then,
\begin{equation}
\frac{F(X)}{X\wedge1}\ge_*\frac{g(X)}{X\wedge1}
\quad\text{if and only if}\quad
F(j_i)\ge g(j_i)
\quad
i=1,\ldots,I+1
\end{equation}
\end{lemma}

\begin{proof}
If $F(j_i)<g(j_i)-\varepsilon$ for some $\varepsilon>0$ and $i=1,\ldots,I+1$
then by continuity there exists $(j_i-j_{i-1})/2>\eta>0$ such that $f<g-\varepsilon$ 
in restriction to the set $A_i=\{k_i-\eta<X\le k_i\}$. By Assumption \ref{ass ATM}, 
$A_i\notin\Neg_*$, moreover, $X\wedge1\ge (j_i+j_{i-1})/2>0$ on $A_i$ so that 
$F(X)/(X\wedge1)\ge_* g(X)/(X\wedge1)$ is contradicted.

Conversely, if $F(j_i)\ge g(j_i)$ holds for $j=1,\ldots,I+1$, then, given that 
$F(0)=g(0)=0$, that $g$ is convex and $f$ piecewise linear, we conclude that 
$F(x)\ge g(x)$ for all $0\le x\le X^*$ and so that $f(X)\ge g(X)$ and thus 
$F(X)/(X\wedge1)\ge_* g(X)/(X\wedge1)$.
\end{proof}

\begin{lemma}
\label{lemma g}
Let Assumption \ref{ass ATM} hold, choose $g\in\Gamma$. Write 
\begin{align}
\mathbf g
&=
\left[
\begin{tabular}{c}
$g(j_1)$\\
$g(j_2)$\\
$\vdots$\\
$g(j_I)$\\
$g(j_{I+1})$
\end{tabular}
\right],
&
\mathbf D
&=
\left[
\begin{tabular}{cccc}
$(j_1-j_0)$&$0$&$\ldots$&$0$\\
$(j_2-j_0)$&$(j_2-j_1)$&$\ldots$&$0$\\
$\vdots$&$\vdots$&$\ddots$&$\vdots$\\
$(j_I-j_0)$&$(j_I-j_1)$&$\ldots$&$0$\\
$(j_{I+1}-j_0)$&$(j_{I+1}-j_1)$&$\ldots$&$(j_{I+1}-j_I)$
\end{tabular}
\right]
\quad\text{and}\quad
\mathbf w=\mathbf D^{-1}\mathbf g
\end{align}
The program
\begin{equation}
\label{qX(f)}
\min_{\{\lambda\theta:\theta\in\Theta_X,\ \lambda>0\}}\lambda q(\theta)
\text{ subject to }
\lambda\frac{X(\theta)}{X\wedge1}\ge_*\frac{g(X)}{X\wedge1}
\end{equation}
is solved by the vector 
\begin{equation}
\label{thetaX(g)}
\theta_X(g)=\sum_{i=0}^I\mathbf w[i+1]\delta_X(j_i)\in\Theta_X
\end{equation}

\end{lemma}

\begin{proof}
It is, first of all, clear that in solving \eqref{qX(f)} one may restrict attention to 
portfolios formed with options with strike prices in $J(X)$. This implies that 
each $\lambda X(\theta)$ with $\lambda>0$ and $\theta\in\Theta_X$ in 
\eqref{qX(f)} may be taken to be of the form 
$F_{\mathbf a}(X)=\sum_{i=0}^Ia_i(X-j_i)^+$ so 
that $\lambda q(\theta)=\mathbf q^T\mathbf a$ with 
\begin{equation*}
\mathbf a^T=[a_0,\ldots,a_I],\ 
\mathbf q^T=[q_X(j_0),\ldots,q_X(j_I)]
	\in
\R^{I+1}_+
\end{equation*}
Remark that 
$[F_\mathbf a(j_1),\ldots,F_\mathbf a(j_{I+1)}]^T=\mathbf D\mathbf a$. 
Fix $g\in\Gamma$. By Lemma \ref{lemma un}
$\frac{F_\mathbf a(X)}{X\wedge1}\ge_*\frac{g(X)}{X\wedge1}$ 
is equivalent to $\mathbf D\mathbf a\ge\mathbf g$.

Define the vectors $\mathbf w,\mathbf b\in\R^{I+1}$ implicitly by letting
\begin{equation}
\label{b recursion}
b_Id_I=q_X(j_I)
\quad\text{and}\quad
b_I+\sum_{i=n}^{I-1}b_i
=\frac{q_X(j_n)-q_X(j_{n+1})}{j_{n+1}-j_n}
\qquad n=0,\ldots,I-1
\end{equation}
and
\begin{equation}
\label{w recursion}
\sum_{i=1}^n\mathbf w[i]
	=
\frac{g(j_n)-g(j_{n-1})}{j_n-j_{n-1}}
\qquad
n=1,\ldots,I
\quad\text{and}\quad
\sum_{i=0}^I\mathbf w[i+1]d_i
	=
g(j_{I+1})
\end{equation}
The following properties are easily established by induction: (\textit i) $\mathbf b\ge0$ 
(as $j_0,\ldots,j_I\in J(X)$) (\textit{ii}) $\mathbf w\ge0$ (as $f\in\Gamma$), 
(\textit{iii}) $\mathbf b^T\mathbf D=\mathbf q^T$ and (\textit{iv}) 
$\mathbf w=\mathbf D^{-1}\mathbf g$. But then, if $\lambda>0$ and $\theta\in\Theta_X$
are such that $\lambda\frac{X(\theta)}{X\wedge1}\ge_*\frac{g(X)}{X\wedge1}$
\begin{align*}
\lambda q(\theta)
	\ge
\min_{\left\{\mathbf a\in\R^{I+1}_+:\mathbf D\mathbf a\ge\mathbf f\right\}}%
\mathbf q^T\mathbf a
	=
\min_{\left\{\mathbf a\in\R^{I+1}_+:\mathbf D\mathbf a\ge\mathbf f\right\}}%
\mathbf b^T\mathbf D\mathbf a
	\ge
\mathbf b^T\mathbf g
	=
\mathbf q^T\mathbf w
	=
q(\theta_X(g))
\end{align*}
\end{proof}

\end{document}